\documentclass[9pt,numbers]{sigplanconf}

\usepackage{amsmath}
\usepackage{amssymb}
\usepackage{mathrsfs}
\usepackage{graphicx}
\usepackage{palatino}
\usepackage{mathtools}
\usepackage{amsthm}

\usepackage{caption}
\DeclareCaptionType{copyrightbox}
\usepackage{subcaption}
\usepackage[figure, linesnumbered]{algorithm2e}

\usepackage{tikz}
\usetikzlibrary{arrows,automata}
\usetikzlibrary{decorations.pathmorphing}
\usetikzlibrary{decorations.pathreplacing}
\usetikzlibrary{fadings}
\usetikzlibrary{patterns}
\tikzset{
    state/.style={
           rectangle,
           rounded corners,
           draw=black,
           minimum height=2em,
           inner sep=2pt,
           text centered,
           },
}

\newtheorem{proposition}{Proposition}

\newtheorem{definition}{Definition}

\DeclareMathOperator{\lin}{lin}

\SetKwBlock{SubAlgoBlock}{}{end}
\newcommand{\SubAlgo}[2]{#1 \SubAlgoBlock{#2}}

\definecolor{gris}{RGB}{180,180,200}
\definecolor{allEvtColor}{RGB}{70,90,230}
\colorlet{allEvtColorFill} {allEvtColor!20!white}
\definecolor{writeColor}{RGB}{240,90,70}
\definecolor{writeColorFill}{RGB}{255,230,200}

\copyrightyear{2016}
\conferenceinfo{PPoPP '16}{March 12-16, 2016, Barcelona, Spain}
\copyrightdata{978-1-4503-4092-2/16/03}
\reprintprice{\$15.00}

\copyrightdoi{2851141.2851170} % ACM e-mail had
\publicationrights{licensed}     % if ACM e-mail had
\begin{document}

%\titlebanner{banner above paper title}        % These are ignored unless
%\preprintfooter{short description of paper}   % 'preprint' option specified.

\title{Causal Consistency: Beyond Memory}

\authorinfo{Matthieu Perrin \and Achour Mostefaoui \and Claude Jard}
           {LINA -- University of Nantes, Nantes, France}
           {[firstname.lastname]@univ-nantes.fr}

\maketitle

\begin{abstract}
In distributed systems where strong consistency is costly when not impossible, causal consistency provides a valuable abstraction to represent program executions as partial orders. In addition to the sequential program order of each computing entity, causal order also contains the semantic links between the events that affect the shared objects -- messages emission and reception in a communication channel, reads and writes on a shared register. Usual approaches based on semantic links are very difficult to adapt to other data types such as queues or counters because they require a specific analysis of causal dependencies for each data type. This paper presents a new approach to define causal consistency for any abstract data type based on sequential specifications. It explores, formalizes and studies the differences between three variations of causal consistency and highlights them in the light of PRAM, eventual consistency and sequential consistency: weak causal consistency, that captures the notion of causality preservation when focusing on convergence; causal convergence that mixes weak causal consistency and convergence; and causal consistency, that coincides with causal memory when applied to shared memory.
\end{abstract}



% ============================================================================
\section{Introduction}
% ============================================================================

\paragraph{Overview.}
Distributed systems are often viewed as more difficult to program than sequential systems 
because they require to solve many issues related to communication. Shared objects, that can be
accessed concurrently by the processes of the system, can be used as a practical abstraction
of communication to let processes enjoy a more general view of the system and tend to meet the classical paradigm of parallel programming. A precise specification of these objects is essential to ensure their adoption as well as the reliability issues of distributed systems. The same reasoning also holds for parallel and multicore processors \cite{BA08,SVNJS11}.

Many models have been proposed to specify shared memory \cite{adve1996shared}.
Linearizability \cite{herlihy1990linearizability} and sequential consistency \cite{lamport1979make} 
guarantee that all the operations appear totally ordered, and that this order is compatible with
the \emph{program order}, the order in which each process performs its own operations.
These strong consistency criteria are very expensive to implement in message-passing systems.
In terms of time, the duration of either the read or the write operations has to be linear with the latency of the network for sequential consistency \cite{lipton1988pram} and for all kind of operations in the case of linearizability \cite{attiya1994sequential}.
Concerning fault-tolerance, strong hypotheses must be respected by the system: 
it is impossible to resist partitioning (CAP Theorem) \cite{gilbert2002brewer}.

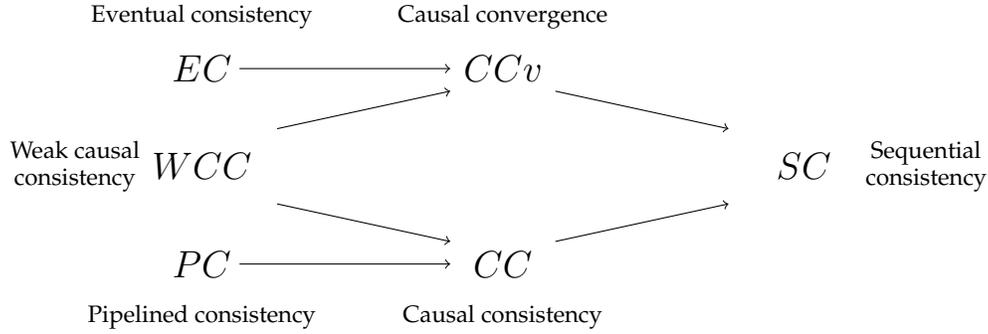
\begin{figure*}
  \centering
    \begin{tikzpicture}
      \draw     (3,3.3)  node{\Large $EC$} ;
      \draw     (3,4)  node{Eventual consistency} ;

      \draw     (7,3.3) node{\Large $CCv$} ;
      \draw     (7,4) node{Causal convergence} ;
      
      \draw     (2.5,2)  node[left]{\begin{tabular}{c}Weak causal \\consistency\end{tabular}} ;
      \draw     (3,2)  node[]{\Large $WCC$} ;
      \draw     (11,2)    node{\Large $SC$} ;
      \draw     (11.5,2)    node[right]{\begin{tabular}{c}Sequential \\consistency\end{tabular}} ;
      
      \draw     (3,0.7)  node{\Large $PC$} ;
      \draw     (3,0)  node{Pipelined consistency} ;

      \draw     (7,0.7)  node{\Large $CC$} ;
      \draw     (7,0)  node{Causal consistency} ;
      
      \draw[->] (3.5,3.3)      -- (6.3,3.3) ;   % EC->CCv
      \draw[->] (4,2.5)  -- (6.3,3) ; % WCC->CCv
      \draw[->] (7.7,3) -- (10,2.5) ;  % CCv->SC

      \draw[->] (3.5,0.7)      -- (6.3,0.7) ; % PC->CC
      \draw[->] (4,1.5)    -- (6.3,1) ;     % WCC->CC
      \draw[->] (7.7,1)    -- (10,1.5) ;     % CC->SC
    \end{tikzpicture}
  \caption{Relative strength of causality criteria}
  \label{fig:sum_up_criteria}
\end{figure*}

In order to gain in efficiency, researchers explored weak consistency criteria, especially for parallel machine in the 
nineties like PRAM \cite{lipton1988pram} and causal memory \cite{ahamad1995causal} that are the best documented. 
PRAM ensures that each process individually sees a consistent local history, that respects the order in which the 
other processes performed their own writes on a register (small piece of memory). Recently, we have seen a resurgence of interest for this topic 
due to the development of multicore processors \cite{H09} and cloud computing (large modern distributed systems such 
as Amazon's cloud, data centers \cite{vogels2008eventually}) from the one side and the necessity to circumvent the CAP 
impossibility result and to ensure high efficiency.
Among recent works, CRDT distributed data types \cite{SPBZ11} and a parallel specification of eventual consistency \cite{BGYZ14}. 
Eventual consistency ensures that all the processes will eventually reach a common state when they stop writing.

{\em Causality} models a distributed execution as a partial order. It was first defined for message-passing
systems, on top of Lamport's \emph{happens-before} relation \cite{lamport1978time}, 
a partial order that contains the sequential program order of each process, and in which a message
emission happens-before its reception. Causal reception \cite{birman1987reliable,raynal1991causal}
ensures that, if the emission of two messages sent to the same process are related by the
happens-before relation, their reception will happen in the same order.
In shared memory models, the exchange of information is ensured by
shared registers where processes write and read information. 
Causal memory \cite{ahamad1995causal} aims at building a causal order, 
that also contains the sequential program order of the processes, 
but in which the emission and reception of messages are replaced by
read and write operations on shared registers. 
The guarantees that are ensured by causal consistency have been identified as
four \emph{session guarantees} by Terry et. al in \cite{terry1994session}. 
\emph{Read your writes} ensures that a read cannot return a value older than a
value written earlier by the same process. 
\emph{Monotonic writes} ensures that, if a process writes into two different registers, 
and another process reads the second write, then this process must read a value at 
least as recent as the first write in the other register.
\emph{Monotonic reads} ensures that, if a process writes twice the same registers, 
and another process reads the second writes, then this process can never read the 
first value again in the future.
Finally, \emph{writes follow reads} ensures that, if a process reads a value in a register
and then writes a value in another register, another process cannot read the lastly written
value and a value older than the value read by the first process. 

\paragraph{Motivation.}

On the one hand, strong consistency criteria (linearizability and sequential consistency) are 
costly both in time and space and turn out to be impossible to implement in some systems. On the other 
hand, there are many weak consistency criteria (e.g. eventual consistency, PRAM consistency, causal 
consistency) that can be implemented in any distributed system where communication is possible 
and with a time complexity that does not depend on communication delays. The natural question is 
then "what is the strongest consistency criterion" that enjoys such property? Unfortunately, 
it has been proven in \cite{perrin2015update} that when wait-free distributed systems are 
considered (all but one process may crash), PRAM consistency (or, \textit{a fortiori}, causal 
consistency) and eventual consistency cannot be provided together.
This means that there are at least two separate branches in the hierarchy of weak consistency criteria: 
one that contains PRAM and causal consistency and the other that contains eventual consistency. As a 
consequence, a question arises "what are the causality properties that can be provided together with 
eventual consistency?". 

Causal consistency, as is known up till now, is only defined for memory (registers) and it has been defined 
assuming it is implemented using messages. Indeed, memory is a good abstraction in sequential programming 
models. Things are more complicated for distributed computing because of race conditions: complex 
concurrent editing can often lead to inconsistent states. Critical sections and transactions offer 
generic solutions to this problem, but at a high cost: they reduce parallelism and fault-tolerance and 
may lead to deadlocks. Another approach is to design the shared objects directly, without using shared 
memory. For example, in the context of collaborative editing, the CCI model \cite{sun1998achieving} 
requires convergence, causality and intention preservation. In this model, causality preservation is 
weaker than causal consistency as defined for memory and can be provided together with eventual consistency.

The definition of causal memory is based on a semantic matching between the reads and the writes. 
In particular, a read operation depends only on the last write operation. For other abstract data 
types (e.g. graphs, counters or queues) the value returned by a query does not depend on one particular 
update, but on all or part of the updates that happened before it. Moreover, for a queue, the order 
of these updates is important. It would be interesting to define causal consistency for any object 
having a sequential specification independently from any implementation mechanism.

\paragraph{Contributions of the paper.}
As said above, this paper aims at extending the definition of causal consistency to all abstract data types. 
To this end, it introduces a clear distinction between two facets that are necessary
to fully specify shared objects: a sequential specification using state transition automata, and a consistency criterion,
that defines a link between (distributed) histories and sequential specifications. 
Its main contribution is the formal definition of three variations of causal 
consistency that can be provided in any distributed system independently
from communication delays (network latency) meaning that an operation returns without waiting any contribution 
from other processes. These criteria complete and help to better understand the map (Fig. \ref{fig:sum_up_criteria}) 
that gives an overview of the relative strength of consistency
criteria.
An arrow from a criterion $C_1$ to a criterion $C_2$ in Fig. \ref{fig:sum_up_criteria} denotes the fact that $C_2$ is stronger that $C_1$.
\begin{itemize}
  \item Weak causal consistency. It can be seen as the causal common denominator between the two branches 
    and can be associated with any weak consistency criteria to form a new criterion that can be implemented 
    in a wait-free distributed system.
  \item Causal convergence. It is the combination of eventual consistency and weak causal consistency. 
    This can be a candidate to replace eventual consistency.
  \item Causal consistency. When applied to registers, it matches the definition of causal memory. 
\end{itemize}

In order to illustrate the notions presented in this paper, a data structure called window stream of size $k$ is introduced. 
This data structure allows to capture the diversity of data structures thanks to the parameter $k$.

The remainder of this paper is organized as follows. Section \ref{section:formalization}
presents a formalization of abstract data types as well as the notion of consistency criteria.
Section \ref{section:weak_causal_consistency} defines and illustrates weak causal consistency.
Section \ref{sec:cc} defines and illustrates causal consistency and compares it to causal memory.
Section \ref{sec:ccv} defines and illustrates causal convergence.
Section \ref{section:implementation} discusses the implementation of causal consistency and causal 
convergence in asynchronous message-passing distributed systems where crashes may occur.
Finally, Section \ref{section:conclusion} concludes the paper.

% ============================================================================
\section{Specifying shared objects}\label{section:formalization}
\vspace{2mm}
% ============================================================================

Shared objects can be specified by two complementary facets:
an \emph{abstract data type} that has a sequential specification, defined in this 
paper by a transition system that characterizes the sequential histories allowed 
for this object and a \emph{consistency criterion}
that makes the link between the sequential specifications and the
distributed executions that invoke them, by a characterization of the histories that
are admissible for a program that uses the objects, depending on their type.
As shared objects are implemented in a distributed system, typically using replication, 
the events in a distributed history are partially ordered. 

\vspace{1mm}
\subsection{Abstract data types}
\vspace{1mm}

To our knowledge, the only attempt to define weakly consistent objects for arbitrary abstract data types is based on 
parallel specifications \cite{BGYZ14}, in which the state accessed by a process at each operation is defined by a 
function on the operations in its past, ordered by a \emph{visibility} and an \emph{arbitration} relations.
The first limit of this approach is that parallel specifications are, by design, only suitable to express 
strong eventually consistent objects, in which two processes must see the same state as soon as they have received 
the same updates. The second limit is that parallel specifications require to specify a state for any possible 
partial order of events, which leads to specifications as complex as the programs they specify.
Consequently, they are non-intuitive and error-prone
as they cannot rely on the well-studied and understood notions of abstract states and transitions.

We use transition systems to specify sequential abstract data types. 
We model abstract data types as transducers, very close to Mealy machines 
\cite{mealy1955method}, except that we do not restrict the analysis
to finite state systems. The input alphabet $\Sigma_i$ consists of
the set of the methods available on the data type. Each method can have two effects. 
On the one hand, they can have a side effect that usually affects all processes. In the transition system, 
it corresponds to a transition between abstract states formalized by the transition function $\delta$. 
On the other hand, they can return a value from the output alphabet $\Sigma_o$ depending on the abstract state
and the output function $\lambda$. Both the transition and the output functions must be total, as shared objects
evolve according to external calls to their operations, to which they must respond in all circumstances.
For example, the \emph{pop} method from a stack deletes the
head of the stack (the side effect) and returns its value (the output). 
More formally, abstract data types correspond to Def. \ref{def:adt}.

\vspace{1mm}
\begin{definition}\label{def:adt}
  An \emph{abstract data type} (ADT) is a 6-tuple \linebreak $T = (\Sigma_i, \Sigma_o, Q, q_0, \delta, \lambda)$ such that:
  \begin{itemize}
  \item\vspace{2mm}  $\Sigma_i$ and $\Sigma_o$ are countable sets called \emph{input} and \emph{output} alphabets. 
    An \emph{operation} is an element $(\sigma_i, \sigma_o)$ of  $\Sigma_i \times \Sigma_o$, denoted by $\sigma_i/\sigma_o$;
  \item\vspace{2mm}  $Q$ is a countable set of \emph{states} and $q_0 \in Q$ is the \emph{initial state};
  \item\vspace{2mm}  $\delta :  Q\times \Sigma_i\rightarrow Q$ and $\lambda : Q\times \Sigma_i\rightarrow \Sigma_o$ are the \emph{transition} and \emph{output} functions.
  \end{itemize}
\end{definition}
\vspace{1mm}

We distinguish two kinds of operations depending on their behavior: updates and queries. 
An input $\sigma_i$ is an \emph{update} if the transition part is not always a loop, i.e. 
there is a state $q$ such that $\delta(q,\sigma_i)\neq q$.
It is a \emph{query} if the output depends on the state, i.e. there are two states $q$ and $q'$
such that $\lambda(q,\sigma_i)\neq\lambda(q',\sigma_i)$.
Some operations are both update and query. For example, the pop operation in a stack 
deletes the first element (the update part) and returns its value (the query part).
An operation that is not an update (resp. query) is called a \emph{pure query} 
(resp. \emph{pure update}).

\paragraph{Sequential specification.}
We now define the sequential specification $L(T)$ of an ADT $T$. A sequential specification is a set 
of sequences of operations that label paths in the transition system, starting from the initial state. 
We need to take into account two additional features in our model: prefixation and hidden operations (Def. \ref{def:L(T)}).
\begin{itemize}
\item We need to take into consideration both finite and infinite sequences. To do so, we first define infinite
  sequences recognized by $T$, and we then extend the concept to the finite prefixes of these sequences.
\item In weak consistency criteria defined on memory, and especially in causal memory, reads and writes 
  usually play a different role. To extend these concepts to generic ADTs where some operations are both
  an update and a query, we need a way to express the fact that the side effect of an operation must be 
  taken into account, but not its return value. To do so, we introduce the notion of \emph{hidden operations},
  in which the method called is known, but not the returned value. Thus, sequential histories admissible 
  for $T$ are sequences of elements of $\Sigma = (\Sigma_i\times \Sigma_o) \cup \Sigma_i$: each element of $\Sigma$ is 
  either an operation $\sigma_i/\sigma_o \in (\Sigma_i\times \Sigma_o)$ or a hidden operation $\sigma_i \in \Sigma_i$.
\end{itemize}

\begin{definition}\label{def:L(T)}
  Let $T = (\Sigma_i, \Sigma_o, Q, q_0, \delta, \lambda)$ be an abstract data type.

  An infinite sequence $(\sigma_i^k/\sigma_o^k)_{k\in \mathbb{N}}$ of operations is \emph{recognized by $T$} if 
  there exists an infinite sequence of states $(q^k)_{k\in \mathbb{N}}$ such that $q^0 = q_0$ is the initial state 
  and for all $k\in \mathbb{N}$, $\delta(q^k, \sigma_i^k) = q^{k+1}$ and $\lambda(q^k, \sigma_i^k) = \sigma_o^k$.

  A finite or infinite sequence $u = (u^k)_{k\in D}$ where $D$ is either $\mathbb{N}$ or $\{0, ..., |u|-1\}$ is a 
  \emph{sequential history admissible for $T$} if there exists an infinite sequence $(\sigma_i^k/\sigma_o^k)_{k\in \mathbb{N}}$ of operations recognized by $T$
  such that, for all $k\in D$, $u^k = \sigma_i^k/\sigma_o^k$ or $u^k = \sigma_i^k$.
  
  The set of all sequential histories admissible for $T$, denoted by $L(T)$, is called the \emph{sequential specification} of $T$.
\end{definition}

\paragraph{Window stream data type.}
Causal consistency has been defined only for memory. The memory abstract data type is very restrictive as a write 
on a register erases the complete past of all the previously written values on that register. For more complex objects 
like stacks and queues, the value returned by a query may depend on more than one update operation, and the order 
in which the different updates were done is important. To illustrate our work on consistency criteria, we need, as a guideline example,
a data type with a simple specification and whose behaviour shows all these features.

We thus introduce the \emph{window stream} data type. In short, it can be seen as a generalization of a register in the sense that the read operation returns the sequence of the last written values instead of the very last one. A window stream of size $k$ (noted $\mathcal{W}_k)$ can be accessed by a \emph{write} operation
$w(v)$ where $v\in \mathbb{N}$ is the written value and a \emph{read} operation $r$ that returns 
the sequence of the last $k$ written values. Missing values are replaced by the 
default value $0$ (any different default value can be considered). More formally, a window stream corresponds to the ADT given in Def. \ref{def:window_reg}.

\pagebreak
\begin{definition}\label{def:window_reg} An integer \emph{window stream} of size $k$ ($k\in \mathbb{N}$) is an ADT
  $\mathcal{W}_k =  (\Sigma_i, \Sigma_o, Q, q_0, \delta, \lambda)$
  with  $\Sigma_i = \cup_{v\in \mathbb{N}}\{r, w(v)\}$, $\Sigma_o = \mathbb{N}^k \cup \{\bot\}$, $Q = \mathbb{N}^k$, $q_0 = (0, ..., 0)$ and,
  for all $v\in \mathbb{N}$ and $q = (q_1, ..., q_k) \in Q$, $\delta(q, w(v)) = (q_2, ...,  q_k, v)$, $\delta(q, r) = q$, 
  $\lambda(q, w(v)) = \bot$ and $\lambda(q, r) = q$.
\end{definition}

The window stream data type has also a great interest in the classification of synchronization objects. The notion of consensus number has been introduced in \cite{herlihy1991wait} to rank the synchronization power of objects. An object has a consensus number equal to $c$ if it allows to reach consensus among $c$ processes and not among $c+1$ processes. Recall that a consensus object can be invoked by a given number of processes. Each process invokes it once with a proposed value and gets a return value such that the returned value has been proposed by some process and all invoking processes obtain the same return value. An object has a consensus number of $c$ if it can emulate a consensus object that can be invoked by at most $c$ processes. As an example the consensus number of a register and a stack are respectively 1 and 2 while the well-known synchronization object compare-and-swap allows to reach consensus among any number of processes. It is interesting to note that a window stream of size $k$ has a consensus number of $k$: if $k$ processes write their proposed values in a sequentially consistent window stream and then return the oldest written value (different from the default value), they will all return the same value. Consequently, a window stream of size at least 2 cannot be implemented using any number of registers (window streams of size 1). 

Additional examples with queues and the complete definition of the memory ADT are given in section \ref{sec:cc}.

\subsection{Distributed histories}

During the execution of a distributed program, the participants/processes call methods on
shared objects (registers, stacks, queues, etc.), an object being an instance of an abstract
data type. An event is the execution of a method by a process. Thereby, each event is labelled 
by an operation from a set $\Sigma$, that usually contains the same symbols as the alphabet 
of the sequential specification $L(T)$.

In a distributed system composed of communicating sequential processes,
all the events produced by one process are totally ordered according to the \emph{program order},
while two events produced by different processes may be incomparable according to the program order. 
In this model, the partially ordered set of events is a collection of disjoint maximal chains: 
each maximal chain of the history corresponds to the events of one process. 
We identify the processes and the events they produce, calling a maximal chain in the history a "process".

Parallel sequences of maximal chains are a too restrictive model to encode the complex behaviour of many distributed systems,
such as multithreaded programs in which threads can fork and join, Web services orchestrations, sensor networks, etc.
Instead, we allow the program order to be any partial order in which all events have a finite past.
In this general model, an event can be contained is several maximal chains. This causes no problems in our definitions.  

\begin{definition}
  A distributed history (or simply history) is a 4-tuple $H = (\Sigma, E, \Lambda, \mapsto)$ such that:
 $\Sigma$ is a countable sets of \emph{operations} in the form $\sigma_i/\sigma_o$ or $\sigma_i$;
 $E$ is a countable set of \emph{events} (denoted by $E_H$ for any history $H$);
 $\Lambda : E \rightarrow \Sigma$ is a \emph{labelling function};
 ${\mapsto} \subset {(E\times E)}$ is a partial order called \emph{program order}, such that
  each event $e\in E$ has a finite past $\{e'\in E : e' \mapsto e\}$. 
\end{definition}

Let $H = (\Sigma, E, \Lambda, \mapsto)$ be a distributed history. Let us introduce a few notations.
\begin{itemize}
\item $\mathscr{P}_H$ denotes the set of the maximal chains of $H$, i.e. maximal totally-ordered sets of events.
  In the case of sequential processes, each $p\in \mathscr{P}_H$ corresponds to the events produced by a process.
  In the remainder of this article, we use the term "process" to designate such a chain, even in models that are not based on a 
  collection of communicating sequential processes.
\item A linearization of $H$ is a sequential history that contains the events of $H$ 
  in an order consistent with the program order. More precisely, it is a word 
  $\Lambda(e_0)\ldots\Lambda(e_i)\ldots$ such that $\{e_0, \ldots, e_i, \ldots\} = E_H$
  and for all $i<j$, $e_j\not\mapsto e_i$.
  $\lin(H)$ denotes the set of all linearizations of $H$.
\item We also define a projection operator $p$ that
  removes part of the information of the history. For $E', E''\subset E$, 
  $H.\pi(E', E'')$ only keeps the operations that are in $E'$, and hides the 
  output of the events that are not in $E''$: 
  $H.\pi(E', E'') = (\Sigma, E', \Lambda', {\mapsto} \cap {E'^2})$
  with 
  \begin{itemize}
  \item $\Lambda'(e) = \sigma_i$ if $\Lambda(e) = \sigma_i/\sigma_o$ and $e\not \in E''$
  \item $\Lambda'(e) = \Lambda(e)$ otherwise.
  \end{itemize}
  Considering memory, $H.\pi(E', E'')$ contains the writes of $E'$ and the reads of $E' \cap E''$.
\item Finally, we define a projection on the histories to replace the program order 
  by another order $\rightarrow$: if $\rightarrow$ respects the 
  definition of a program order (i.e. all events have a finite past in $\rightarrow$), 
  $H^\rightarrow = (\Sigma, E, \Lambda, \rightarrow)$ is the history that contains the same events as 
  $H$, but ordered according to $\rightarrow$.
\end{itemize}

Note that the discreteness of the space of the events does not mean 
that the operations must return immediately, as our model does not 
introduce any notion of real time.

\subsection{Consistency criteria}

A consistency criterion characterizes which histories are admissible for a given data type. 
Graphically, we can imagine a consistency criterion as a way to take a picture of the 
distributed histories so that they look sequential.
More formally, it is a function $C$ that associates a set of consistent histories $C(T)$ 
with any ADT $T$. An implementation of a shared object is $C$-consistent for a consistency
criterion $C$ and an ADT $T$ if all the histories it admits are in $C(T)$.
For the sake of clarity, we will define consistency criteria by a predicate $P(T, H)$ 
that depends on an ADT $T$ and a distributed history $H$. A criterion is defined as 
the function that associates to each $T$, the set of all the histories $H$ such that $P(T, H)$ is true.

We say that a criterion $C_1$ is \emph{stronger} than a criterion $C_2$ if for any ADT $T$, 
$C_1(T) \subset C_2(T)$. A strong consistency criterion guarantees stronger properties on the 
histories it admits. Hence, a $C_1$-consistent implementation can always be used instead of a 
$C_2$-consistent implementation of the same abstract data type if $C_1$ is stronger than $C_2$.
We now define sequential consistency \cite{lamport1979make} 
and pipelined consistency \cite{lipton1988pram} to illustrate this formalism.

\paragraph{Sequential consistency.} was originally defined by Lamport 
in \cite{lamport1979make}: {\em the result of any execution is the same as if the operations 
of all the processors were executed in some sequential order, and the operations of each 
individual processor appear in this sequence in the order specified by its program}. 
In our formalism, such a sequence is a word of operations that has two properties:
it is correct with respect to the sequential specification of the object (i.e. it belongs to $L(T)$)
and the total order is compatible with the program order (i.e. it belongs to $\lin(H)$).

\begin{definition}
  A history $H$ is \emph{sequentially consistent} (SC) with an ADT $T$ if: $\lin(H) \cap L(T) \neq \emptyset.$
\end{definition}

\paragraph{Pipelined consistency.} The PRAM consistency criterion (for "Pipelined Random Access Memory") 
has been defined for shared memory \cite{lipton1988pram}. In PRAM consistency, the processes only have 
a partial view of the history.
More precisely, they are aware of their own reads and all the writes. PRAM consistency ensures that
the view of each process is consistent with the order in which the writes were made by each process.
Each process must be able to explain the history by a linearization of its own knowledge. 
Pipelined consistency is weaker than sequential consistency, for
which it is additionally required that the linearizations seen by different processes be identical.
The PRAM consistency is local to each process. As different processes can see concurrent 
updates in a different order, the values of the registers do not necessarily converge.

\emph{Pipelined consistency} is an extension of PRAM consistency to other abstract data types. 
As not all operations are either pure updates or pure queries, we use the projection operator to 
hide the return values (the output alphabet) of all the events that are not made by a process.
For each process $p$, $H.\pi(E_H, p)$ is the history that contains all the events of $p$ unchanged,
and the return values of the operations labelling the events of the other processes are unknown.
Pipelined consistency corresponds to Def. \ref{def:PC}. 

\begin{definition}
  \label{def:PC}
  $H$ is \emph{pipelined consistent} (PC) with $T$ if:\linebreak
  $\forall p\in \mathscr{P}_H, \lin\left(H.\pi(E_H, p)\right) \cap L(T) \neq \emptyset.$
\end{definition}

% ============================================================================
\section{Weak causal consistency}\label{section:weak_causal_consistency}
% ============================================================================

\subsection{Causal orders and time zones}

Causal consistency is based on the thought that a distributed system is depicted
by a partial order that represents a \emph{logical time} in which the processes
evolve at their own pace. This partial order, called \emph{causal order}, contains 
the sequential arrangement imposed by the processes. Additionally, an event cannot 
be totally ignored by a process forever (see Def. \ref{def:causal_order}), which 
corresponds to the eventual reception in message-passing systems. There are three reasons why 
cofiniteness is important in our model.
\begin{enumerate}
\item For infinite histories, cofiniteness usually prevents the causal order to be the program order. 
  If we did not impose this restriction, the obtained criteria would be much weaker, as it would not force the processes to 
  interact at all. Such criteria could be implemented trivially, each process updating its own local 
  variable. However, they would not be so useful in distributed systems. 
\item It is usually stated that causal memory is stronger than PRAM. From 
  Def. \ref{def:PC}, the operation associated with each event stands at some finite 
  position in the linearization required for each process. Thus, a criterion in which 
  processes are not required to communicate would not strenghten pipelined consistency.
\item It is also important to ensure that causal convergence is stronger than eventual consistency:
  convergence can only be achieved when all processes have the same updates in their causal past;
  to strenghten eventual consistency, we must ensure that, if all processes stop updating then, eventually,
  all processes will have all the updates in their causal past.
\end{enumerate}
\begin{definition}\label{def:causal_order}
  Let $H$ be a distributed history.
  A \emph{causal order} is a partial order $\rightarrow$ on all the events of $E_H$, 
  that contains $\mapsto$, and such that for all $e\in E_H$, 
  $\{e'\in E_H : e\not\rightarrow e'\}$ is finite.
\end{definition}

In a distributed history augmented with a causal order, for each event $e$, the history
can be divided into six zones: the causal (resp. program) past that contains the predecessors
of $e$ in the causal (resp. program) order, the causal (resp. program) future that contains the 
successors of $e$ in the causal (resp. program) order, the present that contains only $e$ and 
the concurrent present that contains the events incomparable with $e$ for both orders. These zones 
are depicted in Fig. \ref{fig:CC:cones}. The causal past of $e$ is denoted 
by $\lfloor e\rfloor = \{e'\in E_H : e'\rightarrow e\}$. 

\begin{figure*}
  \centering
  \hspace{\fill}
  \begin{subfigure}[b]{0.49\textwidth}
    \centering
    \scalebox{0.8}{
      \begin{tikzpicture}

        \fill[allEvtColorFill] (-0.5,1.75) -- (-0.5,0.25) -- (3.5,0.25) -- (3.5,1.75) -- cycle;
        \fill[allEvtColorFill] (3.5,0.25) -- (3.5,1.75) -- (5.5,1.75) -- (5.5,0.25) -- cycle;

        \fill[writeColorFill] (-0.5,1.75) -- (3.5,1.75) -- (3.5,3) -- (-0.5,3) -- cycle;
        \fill[pattern=north east lines, pattern color=writeColorFill!30!white] (-0.5,1.75) -- (3.5,1.75) -- (3.5,3) -- (-0.5,3) -- cycle;
        \fill[white, path fading=west] (-0.5,1.75) -- (3.5,1.75) -- (3.5,3) -- (-0.5,3) -- cycle;

        \fill[writeColorFill] (-0.5,0.25) -- (3.5,0.25) -- (3.5,-1) -- (-0.5,-1) -- cycle;
        \fill[pattern=north east lines, pattern color=writeColorFill!30!white] (-0.5,0.25) -- (3.5,0.25) -- (3.5,-1) -- (-0.5,-1) -- cycle;
        \fill[white, path fading=west] (-0.5,0.25) -- (3.5,0.25) -- (3.5,-1) -- (-0.5,-1) -- cycle;

        \draw[-] (3.5,1.75) -- (1.75,3) ;
        \draw[-] (3.5,0.25) -- (1.75,-1) ;
        \draw[-] (5.5,1.75) -- (7.25,3) ;
        \draw[-] (5.5,0.25) -- (7.25,-1) ;
        
        \draw[-] (-0.5,0.25) -- (3.5,0.25) ;
        \draw[-] (3.5,0.25) -- (5.5,0.25) ;
        \draw[-] (5.5,0.25) -- (9.5,0.25) ;
        
        \draw[-] (-0.5,1.75) -- (3.5,1.75) ;
        \draw[-] (3.5,1.75) -- (5.5,1.75) ;
        \draw[-] (5.5,1.75) -- (9.5,1.75) ;
        
        \draw[-] (3.5,0.25) -- (3.5,1.75) ;
        \draw[-] (5.5,0.25) -- (5.5,1.75) ;
                
        \draw[writeColor]      (0.5,2)      node{$\bullet$} ;
        \draw      (0.5,2)      node[above]{$\mathbf{\sigma_i^{2}}{\color{gris}/\sigma_o^{2}}$} ;
        \draw[|->] (0.7,2) -- (3.8,2) ;
        \draw[gris]      (4,2)    node{$\bullet$} ;
        \draw[gris]      (4,2)    node[above]{$\sigma_i^6/\sigma_o^6$} ;
        \draw[|->] (4.2,2) -- (7.1,2) ;
        \draw[gris]      (7.3,2)      node{$\bullet$} ;
        \draw[gris]      (7.3,2)      node[above]{$\sigma_i^{9}/\sigma_o^{9}$} ;
        \draw[|->] (7.5,2) -- (8.5,2) ;
        \draw[gris]      (8.7,2)      node{$\bullet$} ;
        \draw[gris]      (8.7,2)      node[above]{$\sigma_i^{12}/\sigma_o^{12}$} ;
        
        \draw[allEvtColor]      (0.8,1)      node{$\bullet$} ;
        \draw      (0.8,1)      node[below]{$\mathbf{\sigma_i^3/\sigma_o^3}$} ;
        \draw[|->] (1,1) -- (1.8,1) ;
        \draw[allEvtColor]      (2,1)      node{$\bullet$} ;
        \draw      (2,1)      node[below]{$\mathbf{\sigma_i^5/\sigma_o^5}$} ;
        \draw[|->] (2.2,1) -- (4.3,1) ;
        \draw[allEvtColor]      (4.5,1)    node{$\bullet$} ;
        \draw      (4.5,1)    node[below]{$\mathbf{\sigma_i^7/\sigma_o^7}$} ;
        \draw[|->] (4.7,1) -- (7.8,1) ;
        \draw[gris]      (8,1)      node{$\bullet$} ;
        \draw[gris]      (8,1)      node[below]{$\sigma_i^{10}/\sigma_o^{10}$} ;
        
        \draw[writeColor]      (0.3,0)      node{$\bullet$} ;
        \draw      (0.3,0)      node[below]{$\mathbf{\sigma_i^{1}}{\color{gris}/\sigma_o^{1}}$} ;
        \draw[|->] (0.5,0) -- (1.3,0) ;
        \draw[gris]      (1.5,0)      node{$\bullet$} ;
        \draw[gris]      (1.5,0)      node[below]{$\sigma_i^4/\sigma_o^4$} ;
        \draw[|->] (1.7,0) -- (4.8,0) ;
        \draw[gris]      (5,0)    node{$\bullet$} ;
        \draw[gris]      (5,0)    node[below]{$\sigma_i^8/\sigma_o^8$} ;
        \draw[|->] (5.2,0) -- (8.3,0) ;
        \draw[gris]      (8.5,0)      node{$\bullet$} ;
        \draw[gris]      (8.5,0)      node[below]{$\sigma_i^{11}/\sigma_o^{11}$} ;

        \draw      (-0.5,2.75)   node[right]{causal past} ;
        \draw      (-0.5,1.45)   node[right]{program past} ;
        \draw      (-0.5,-0.75)   node[right]{causal past} ;
        
        \draw      (4.5,2.75)    node{concurrent present} ;
        \draw      (4.5,1.45)    node{present} ;
        \draw      (4.5,-0.75)    node{concurrent present} ;
        
        \draw      (9.5,2.75)    node[left]{causal future} ;
        \draw      (9.5,1.45)    node[left]{program future} ;
        \draw      (9.5,-0.75)    node[left]{causal future} ;

      \end{tikzpicture}
    }
    \caption{Pipelined consistency.}
    \label{subfig:CC:cones:PC}
  \end{subfigure}
  \hspace{\fill}
  \begin{subfigure}[b]{0.49\textwidth}
    \centering
    \scalebox{0.8}{
      \begin{tikzpicture}

        \fill[allEvtColorFill] (3.5,0.25) -- (3.5,1.75) -- (5.5,1.75) -- (5.5,0.25) -- cycle;

        \fill[writeColorFill] (-0.5,3) -- (1.75,3) -- (3.5,1.75) -- (3.5,0.25) -- (1.75,-1) -- (-0.5,-1) -- cycle;
        \fill[pattern=north east lines, pattern color=writeColorFill!30!white] (-0.5,3) -- (1.75,3) -- (3.5,1.75) -- (3.5,0.25) -- (1.75,-1) -- (-0.5,-1) -- cycle;
                
        \draw[-] (3.5,1.75) -- (1.75,3) ;
        \draw[-] (3.5,0.25) -- (1.75,-1) ;
        \draw[-] (5.5,1.75) -- (7.25,3) ;
        \draw[-] (5.5,0.25) -- (7.25,-1) ;
        
        \draw[-] (-0.5,0.25) -- (3.5,0.25) ;
        \draw[-] (3.5,0.25) -- (5.5,0.25) ;
        \draw[-] (5.5,0.25) -- (9.5,0.25) ;
        
        \draw[-] (-0.5,1.75) -- (3.5,1.75) ;
        \draw[-] (3.5,1.75) -- (5.5,1.75) ;
        \draw[-] (5.5,1.75) -- (9.5,1.75) ;
        
        \draw[-] (3.5,0.25) -- (3.5,1.75) ;
        \draw[-] (5.5,0.25) -- (5.5,1.75) ;
                
        \draw[writeColor]      (0.5,2)      node{$\bullet$} ;
        \draw      (0.5,2)      node[above]{$\mathbf{\sigma_i^{2}}{\color{gris}/\sigma_o^{2}}$} ;
        \draw[|->] (0.7,2) -- (3.8,2) ;
        \draw[gris]      (4,2)    node{$\bullet$} ;
        \draw[gris]      (4,2)    node[above]{$\sigma_i^6/\sigma_o^6$} ;
        \draw[|->] (4.2,2) -- (7.1,2) ;
        \draw[gris]      (7.3,2)      node{$\bullet$} ;
        \draw[gris]      (7.3,2)      node[above]{$\sigma_i^{9}/\sigma_o^{9}$} ;
        \draw[|->] (7.5,2) -- (8.5,2) ;
        \draw[gris]      (8.7,2)      node{$\bullet$} ;
        \draw[gris]      (8.7,2)      node[above]{$\sigma_i^{12}/\sigma_o^{12}$} ;
        
        \draw[writeColor]      (0.8,1)      node{$\bullet$} ;
        \draw      (0.8,1)      node[below]{$\mathbf{\sigma_i^{3}}{\color{gris}/\sigma_o^{3}}$} ;
        \draw[|->] (1,1) -- (1.8,1) ;
        \draw[writeColor]      (2,1)      node{$\bullet$} ;
        \draw      (2,1)      node[below]{$\mathbf{\sigma_i^{5}}{\color{gris}/\sigma_o^{5}}$} ;
        \draw[|->] (2.2,1) -- (4.3,1) ;
        \draw[allEvtColor]      (4.5,1)    node{$\bullet$} ;
        \draw      (4.5,1)    node[below]{$\mathbf{\sigma_i^7/\sigma_o^7}$} ;
        \draw[|->] (4.7,1) -- (7.8,1) ;
        \draw[gris]      (8,1)      node{$\bullet$} ;
        \draw[gris]      (8,1)      node[below]{$\sigma_i^{10}/\sigma_o^{10}$} ;
        
        \draw[writeColor]      (0.3,0)      node{$\bullet$} ;
        \draw      (0.3,0)      node[below]{$\mathbf{\sigma_i^{1}}{\color{gris}/\sigma_o^{1}}$} ;
        \draw[|->] (0.5,0) -- (1.3,0) ;
        \draw[writeColor]      (1.5,0)      node{$\bullet$} ;
        \draw      (1.5,0)      node[below]{$\mathbf{\sigma_i^{4}}{\color{gris}/\sigma_o^{4}}$} ;
        \draw[|->] (1.7,0) -- (4.8,0) ;
        \draw[gris]      (5,0)    node{$\bullet$} ;
        \draw[gris]      (5,0)    node[below]{$\sigma_i^8/\sigma_o^8$} ;
        \draw[|->] (5.2,0) -- (8.3,0) ;
        \draw[gris]      (8.5,0)      node{$\bullet$} ;
        \draw[gris]      (8.5,0)      node[below]{$\sigma_i^{11}/\sigma_o^{11}$} ;

        \draw      (-0.5,2.75)   node[right]{causal past} ;
        \draw      (-0.5,1.45)   node[right]{program past} ;
        \draw      (-0.5,-0.75)   node[right]{causal past} ;
        
        \draw      (4.5,2.75)    node{concurrent present} ;
        \draw      (4.5,1.45)    node{present} ;
        \draw      (4.5,-0.75)    node{concurrent present} ;
        
        \draw      (9.5,2.75)    node[left]{causal future} ;
        \draw      (9.5,1.45)    node[left]{program future} ;
        \draw      (9.5,-0.75)    node[left]{causal future} ;

      \end{tikzpicture}
    }
    \caption{Weak causal consistency.}
    \label{subfig:CC:cones:WCC}
  \end{subfigure}
  \hspace{\fill}

\vspace{4mm}

  \hspace{\fill}
  \begin{subfigure}[b]{0.49\textwidth}
    \centering
    \scalebox{0.8}{
      \begin{tikzpicture}

        \fill[allEvtColorFill] (-0.5,1.75) -- (-0.5,0.25) -- (3.5,0.25) -- (3.5,1.75) -- cycle;
        \fill[allEvtColorFill] (3.5,0.25) -- (3.5,1.75) -- (5.5,1.75) -- (5.5,0.25) -- cycle;

        \fill[writeColorFill] (-0.5,1.75) -- (3.5,1.75) -- (1.75,3) -- (-0.5,3) -- cycle;
        \fill[pattern=north east lines, pattern color=writeColorFill!30!white] (-0.5,1.75) -- (3.5,1.75) -- (1.75,3) -- (-0.5,3) -- cycle;

        \fill[writeColorFill] (-0.5,0.25) -- (3.5,0.25) -- (1.75,-1) -- (-0.5,-1) -- cycle;
        \fill[pattern=north east lines, pattern color=writeColorFill!30!white] (-0.5,0.25) -- (3.5,0.25) -- (1.75,-1) -- (-0.5,-1) -- cycle;
                
        \draw[-] (3.5,1.75) -- (1.75,3) ;
        \draw[-] (3.5,0.25) -- (1.75,-1) ;
        \draw[-] (5.5,1.75) -- (7.25,3) ;
        \draw[-] (5.5,0.25) -- (7.25,-1) ;
        
        \draw[-] (-0.5,0.25) -- (3.5,0.25) ;
        \draw[-] (3.5,0.25) -- (5.5,0.25) ;
        \draw[-] (5.5,0.25) -- (9.5,0.25) ;
        
        \draw[-] (-0.5,1.75) -- (3.5,1.75) ;
        \draw[-] (3.5,1.75) -- (5.5,1.75) ;
        \draw[-] (5.5,1.75) -- (9.5,1.75) ;
        
        \draw[-] (3.5,0.25) -- (3.5,1.75) ;
        \draw[-] (5.5,0.25) -- (5.5,1.75) ;
                
        \draw[writeColor]      (0.5,2)      node{$\bullet$} ;
        \draw      (0.5,2)      node[above]{$\mathbf{\sigma_i^{2}}{\color{gris}/\sigma_o^{2}}$} ;
        \draw[|->] (0.7,2) -- (3.8,2) ;
        \draw[gris]      (4,2)    node{$\bullet$} ;
        \draw[gris]      (4,2)    node[above]{$\sigma_i^6/\sigma_o^6$} ;
        \draw[|->] (4.2,2) -- (7.1,2) ;
        \draw[gris]      (7.3,2)      node{$\bullet$} ;
        \draw[gris]      (7.3,2)      node[above]{$\sigma_i^{9}/\sigma_o^{9}$} ;
        \draw[|->] (7.5,2) -- (8.5,2) ;
        \draw[gris]      (8.7,2)      node{$\bullet$} ;
        \draw[gris]      (8.7,2)      node[above]{$\sigma_i^{12}/\sigma_o^{12}$} ;
        
        \draw[allEvtColor]      (0.8,1)      node{$\bullet$} ;
        \draw      (0.8,1)      node[below]{$\mathbf{\sigma_i^3/\sigma_o^3}$} ;
        \draw[|->] (1,1) -- (1.8,1) ;
        \draw[allEvtColor]      (2,1)      node{$\bullet$} ;
        \draw      (2,1)      node[below]{$\mathbf{\sigma_i^5/\sigma_o^5}$} ;
        \draw[|->] (2.2,1) -- (4.3,1) ;
        \draw[allEvtColor]      (4.5,1)    node{$\bullet$} ;
        \draw      (4.5,1)    node[below]{$\mathbf{\sigma_i^7/\sigma_o^7}$} ;
        \draw[|->] (4.7,1) -- (7.8,1) ;
        \draw[gris]      (8,1)      node{$\bullet$} ;
        \draw[gris]      (8,1)      node[below]{$\sigma_i^{10}/\sigma_o^{10}$} ;
        
        \draw[writeColor]      (0.3,0)      node{$\bullet$} ;
        \draw      (0.3,0)      node[below]{$\mathbf{\sigma_i^{1}}{\color{gris}/\sigma_o^{1}}$} ;
        \draw[|->] (0.5,0) -- (1.3,0) ;
        \draw[writeColor]      (1.5,0)      node{$\bullet$} ;
        \draw      (1.5,0)      node[below]{$\mathbf{\sigma_i^{4}}{\color{gris}/\sigma_o^{4}}$} ;
        \draw[|->] (1.7,0) -- (4.8,0) ;
        \draw[gris]      (5,0)    node{$\bullet$} ;
        \draw[gris]      (5,0)    node[below]{$\sigma_i^8/\sigma_o^8$} ;
        \draw[|->] (5.2,0) -- (8.3,0) ;
        \draw[gris]      (8.5,0)      node{$\bullet$} ;
        \draw[gris]      (8.5,0)      node[below]{$\sigma_i^{11}/\sigma_o^{11}$} ;

        \draw      (-0.5,2.75)   node[right]{causal past} ;
        \draw      (-0.5,1.45)   node[right]{program past} ;
        \draw      (-0.5,-0.75)   node[right]{causal past} ;
        
        \draw      (4.5,2.75)    node{concurrent present} ;
        \draw      (4.5,1.45)    node{present} ;
        \draw      (4.5,-0.75)    node{concurrent present} ;
        
        \draw      (9.5,2.75)    node[left]{causal future} ;
        \draw      (9.5,1.45)    node[left]{program future} ;
        \draw      (9.5,-0.75)    node[left]{causal future} ;

      \end{tikzpicture}
    }
    \caption{Causal consistency.}
    \label{subfig:CC:cones:CC}
  \end{subfigure}
  \hspace{\fill}
  \begin{subfigure}[b]{0.49\textwidth}
    \centering
    \scalebox{0.8}{
      \begin{tikzpicture}

        \fill[allEvtColorFill] (-0.5,1.75) -- (3.5,1.75) -- (1.75,3) -- (-0.5,3) -- cycle;
        \fill[allEvtColorFill] (-0.5,0.25) -- (3.5,0.25) -- (1.75,-1) -- (-0.5,-1) -- cycle;
        \fill[allEvtColorFill] (-0.5,1.75) -- (-0.5,0.25) -- (3.5,0.25) -- (3.5,1.75) -- cycle;
        \fill[allEvtColorFill] (3.5,0.25) -- (3.5,1.75) -- (5.5,1.75) -- (5.5,0.25) -- cycle;
                
        \draw[-] (3.5,1.75) -- (1.75,3) ;
        \draw[-] (3.5,0.25) -- (1.75,-1) ;
        \draw[-] (5.5,1.75) -- (7.25,3) ;
        \draw[-] (5.5,0.25) -- (7.25,-1) ;
        
        \draw[-] (-0.5,0.25) -- (3.5,0.25) ;
        \draw[-] (3.5,0.25) -- (5.5,0.25) ;
        \draw[-] (5.5,0.25) -- (9.5,0.25) ;
        
        \draw[-] (-0.5,1.75) -- (3.5,1.75) ;
        \draw[-] (3.5,1.75) -- (5.5,1.75) ;
        \draw[-] (5.5,1.75) -- (9.5,1.75) ;
        
        \draw[-] (3.5,0.25) -- (3.5,1.75) ;
        \draw[-] (5.5,0.25) -- (5.5,1.75) ;
                
        \draw[allEvtColor]      (0.5,2)      node{$\bullet$} ;
        \draw      (0.5,2)      node[above]{$\mathbf{\sigma_i^2/\sigma_o^2}$} ;
        \draw[|->] (0.7,2) -- (1.6,2) ;
        \draw[allEvtColor]      (1.8,2)    node{$\bullet$} ;
        \draw      (1.8,2)    node[above]{$\mathbf{\sigma_i^8/\sigma_o^8}$} ;
        \draw[|->] (2,2) -- (7.1,2) ;
        \draw[gris]      (7.3,2)      node{$\bullet$} ;
        \draw[gris]      (7.3,2)      node[above]{$\sigma_i^{9}/\sigma_o^{9}$} ;
        \draw[|->] (7.5,2) -- (8.5,2) ;
        \draw[gris]      (8.7,2)      node{$\bullet$} ;
        \draw[gris]      (8.7,2)      node[above]{$\sigma_i^{12}/\sigma_o^{12}$} ;
        
        \draw[allEvtColor]      (0.8,1)      node{$\bullet$} ;
        \draw      (0.8,1)      node[below]{$\mathbf{\sigma_i^3/\sigma_o^3}$} ;
        \draw[|->] (1,1) -- (1.8,1) ;
        \draw[allEvtColor]      (2,1)      node{$\bullet$} ;
        \draw      (2,1)      node[below]{$\mathbf{\sigma_i^5/\sigma_o^5}$} ;
        \draw[|->] (2.2,1) -- (4.3,1) ;
        \draw[allEvtColor]      (4.5,1)    node{$\bullet$} ;
        \draw      (4.5,1)    node[below]{$\mathbf{\sigma_i^7/\sigma_o^7}$} ;
        \draw[|->] (4.7,1) -- (7.8,1) ;
        \draw[gris]      (8,1)      node{$\bullet$} ;
        \draw[gris]      (8,1)      node[below]{$\sigma_i^{10}/\sigma_o^{10}$} ;
        
        \draw[allEvtColor]      (0.3,0)      node{$\bullet$} ;
        \draw      (0.3,0)      node[below]{$\mathbf{\sigma_i^1/\sigma_o^1}$} ;
        \draw[|->] (0.5,0) -- (1.3,0) ;
        \draw[allEvtColor]      (1.5,0)      node{$\bullet$} ;
        \draw      (1.5,0)      node[below]{$\mathbf{\sigma_i^4/\sigma_o^4}$} ;
        \draw[|->] (1.7,0) -- (6.9,0) ;
        \draw[gris]      (7.1,0)    node{$\bullet$} ;
        \draw[gris]      (7.1,0)    node[below]{$\sigma_i^8/\sigma_o^8$} ;
        \draw[|->] (7.3,0) -- (8.3,0) ;
        \draw[gris]      (8.5,0)      node{$\bullet$} ;
        \draw[gris]      (8.5,0)      node[below]{$\sigma_i^{11}/\sigma_o^{11}$} ;

        \draw      (-0.5,2.75)   node[right]{causal past} ;
        \draw      (-0.5,1.45)   node[right]{program past} ;
        \draw      (-0.5,-0.75)   node[right]{causal past} ;
        
        \draw      (4.5,2.75)    node{concurrent present} ;
        \draw      (4.5,1.45)    node{present} ;
        \draw      (4.5,-0.75)    node{concurrent present} ;
        
        \draw      (9.5,2.75)    node[left]{causal future} ;
        \draw      (9.5,1.45)    node[left]{program future} ;
        \draw      (9.5,-0.75)    node[left]{causal future} ;

      \end{tikzpicture}
    }
    \caption{Sequential consistency.}
    \label{subfig:CC:cones:SC}
  \end{subfigure}
  \hspace{\fill}
  \caption{The differences between causality criteria can be explained in terms of time zones.
    The more constraints the past imposes on the present, the stronger the criterion. 
    The zones in plain blue must be respected totally, and the updates of the zones in striped orange must be taken into account.}
  \label{fig:CC:cones}
\end{figure*}
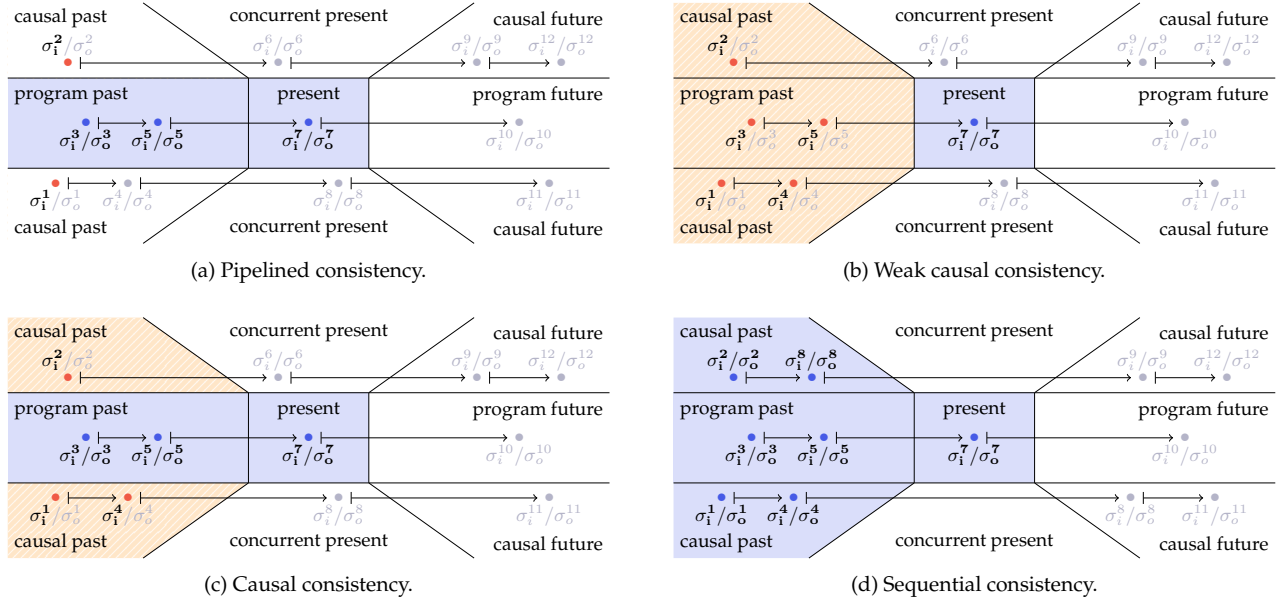

Causal consistency aims at providing a causal order that can be helpful for the
final user when designing an application at a higher level. Causality is not an order imposed by
outer conditions (e.g. the network system), even if causal reception can help in the implementation.
Thus, the existence of a causal order is only required, but not necessarily unique.
An illustration of this point is the fact that no communication is required to insert pure update operations 
into the causal order. 

\begin{figure*}[t]
  \centering
  \hspace{\fill}
  \begin{subfigure}[b]{0.23\textwidth}
    \centering
    \begin{tikzpicture}
      \draw      (1.2,1)    node{$\bullet$} ;
      \draw      (1.2,1)    node[above]{$w(1)$} ;
      \draw[|->] (1.4,1) -- (2.3,1) ;
      \draw      (2.5,1)    node{$\bullet$} ;
      \draw      (2.5,1)    node[above]{$r/(0,1)$} ;
      \draw[|->] (2.7,1) -- (3.6,1) ;
      \draw      (3.8,1)    node{$\bullet$} ;
      \draw      (3.8,1)    node[above]{$r/(1,2)$} ;

      \draw      (1.2,0.5)    node{$\bullet$} ;
      \draw      (1.2,0.5)    node[below]{$w(2)$} ;
      \draw[|->] (1.4,0.5) -- (2.3,0.5) ;
      \draw      (2.5,0.5)    node{$\bullet$} ;
      \draw      (2.5,0.5)    node[below]{$r/(0,2)$} ;
      \draw[|->] (2.7,0.5) -- (3.6,0.5) ;
      \draw      (3.8,0.5)    node{$\bullet$} ;
      \draw      (3.8,0.5)    node[below]{$r/(1,2)$} ;
      
      \draw[->,dashed] (1.4,0.6) -- (3.6,0.9) ;
      \draw[->,dashed] (1.4,0.9) -- (3.6,0.6) ;
    \end{tikzpicture}
    \caption{\footnotesize $\mathcal{W}_2$: CCv, not PC}
    \label{fig:WCC_pas_PC}
  \end{subfigure}
  \hspace{\fill}
  \begin{subfigure}[b]{0.18\textwidth}
    \centering
    \begin{tikzpicture}
      \draw      (2,1)    node{$\bullet$} ;
      \draw      (2,1)    node[above]{$w(1)$} ;

      \draw      (2.5,0.75)    node{$\bullet$} ;
      \draw      (2.5,0.75)    node[below]{$r/(0, 1)$} ;
      \draw[|->] (2.7,0.75) -- (3.3,0.75) ;
      \draw      (3.5,0.75)    node{$\bullet$} ;
      \draw      (3.5,0.75)    node[above]{$w(2)$} ;

      \draw      (4,0.5)    node{$\bullet$} ;
      \draw      (4,0.5)    node[below]{$r/(0, 2)$} ;

      \draw[->,dashed] (2.1,0.95) -- (2.4,0.8) ;
      \draw[->,dashed] (3.6,0.7) -- (3.9,0.55) ;
    \end{tikzpicture}
    \caption{\footnotesize $\mathcal{W}_2$: PC, not WCC}
    \label{fig:PC_pas_WCC}
  \end{subfigure}
  \hspace{\fill}
  \begin{subfigure}[b]{0.17\textwidth}
    \centering
    \begin{tikzpicture}
      \draw      (1,1)      node{$\bullet$} ;
      \draw      (1,1)      node[above]{$w(1)$} ;
      \draw[|->] (1.2,1) -- (2.1,1) ;
      \draw      (2.3,1)    node{$\bullet$} ;
      \draw      (2.3,1)    node[above]{$r/(2,1)$} ;
      
      \draw      (1,0.5)      node{$\bullet$} ;
      \draw      (1,0.5)      node[below]{$w(2)$} ;
      \draw[|->] (1.2,0.5) -- (2.1,0.5) ;
      \draw      (2.3,0.5)    node{$\bullet$} ;
      \draw      (2.3,0.5)    node[below]{$r/(1,2)$} ;

      \draw[->,dashed] (1.2,0.9) -- (2.1,0.6) ;
      \draw[->,dashed] (1.2,0.6) -- (2.1,0.9) ;
    \end{tikzpicture}
    \caption{\footnotesize $\mathcal{W}_2$: CC, not CCv}
    \label{fig:SCC_pas_SC}
  \end{subfigure}
  \hspace{\fill}
  \begin{subfigure}[b]{0.15\textwidth}
    \centering
    \begin{tikzpicture}
      \draw      (1,1)      node{$\bullet$} ;
      \draw      (1,1)      node[above]{$w(1)$} ;
      \draw[|->] (1.2,1) -- (1.9,1) ;
      \draw      (2.1,1)    node{$\bullet$} ;
      \draw      (2.1,1)    node[above]{$r/(0,1)$} ;
      
      \draw      (1,0.5)      node{$\bullet$} ;
      \draw      (1,0.5)      node[below]{$w(2)$} ;
      \draw[|->] (1.2,0.5) -- (1.9,0.5) ;
      \draw      (2.1,0.5)    node{$\bullet$} ;
      \draw      (2.1,0.5)    node[below]{$r/(1,2)$} ;
    \end{tikzpicture}
    \caption{\footnotesize $\mathcal{W}_2$: SC}
    \label{fig:SC}
  \end{subfigure}
  \hspace{\fill}

\vspace{4mm}

  \hspace{\fill}
  \begin{subfigure}[b]{0.28\textwidth}
    \centering
    \begin{tikzpicture}
      \draw      (2.5,1)    node{$\bullet$} ;
      \draw      (2.5,1)    node[above]{$push(1)$} ;
      \draw[|->] (2.7,1) -- (3.05,1) ;
      \draw      (3.25,1)    node{$\bullet$} ;
      \draw      (3.15,1.05)    node[below]{$pop/1$} ;
      \draw[|->] (3.45,1) -- (3.8,1) ;
      \draw      (4,1)      node{$\bullet$} ;
      \draw      (3.9,1)      node[above]{$pop/1$} ;
      \draw[|->] (4.2,1) -- (4.55,1) ;
      \draw      (4.75,1)      node{$\bullet$} ;
      \draw      (5.1,1)      node[above]{$push(3)$} ;

      \draw      (3.75,0.5)      node{$\bullet$} ;
      \draw      (3.75,0.5)      node[below]{$push(2)$} ;
      \draw[|->] (3.95,0.5) -- (4.8,0.5) ;
      \draw      (5,0.5)    node{$\bullet$} ;
      \draw      (5,0.5)    node[below]{$pop/3$} ;
      \draw[|->] (5.2,0.5) -- (5.55,0.5) ;
      \draw      (5.75,0.5)      node{$\bullet$} ;
      \draw      (5.85,0.45)      node[above]{$push(1)$} ;

      \draw[->,dashed] (3.85,0.6) -- (3.9,0.9) ;
      \draw[->,dashed] (4.85,0.9) -- (4.9,0.6) ;
      
    \end{tikzpicture}
    \caption{\footnotesize $\mathcal{Q}$: WCC and PC, not CC}
    \label{fig:file_WCC_PC_pas_CC}
  \end{subfigure}
  \hspace{\fill}
  \begin{subfigure}[b]{0.28\textwidth}
    \centering
    \begin{tikzpicture}
      \draw      (4,1)      node{$\bullet$} ;
      \draw      (3.9,1)      node[above]{$pop/1$} ;
      \draw[|->] (4.2,1) -- (4.8,1) ;
      \draw      (5,1)    node{$\bullet$} ;
      \draw      (5,1)    node[above]{$pop/\bot$} ;

      \draw      (2,0.5)      node{$\bullet$} ;
      \draw      (2,0.5)      node[above]{$push(1)$} ;
      \draw[|->] (2.2,0.5) -- (2.55,0.5) ;
      \draw      (2.75,0.5)    node{$\bullet$} ;
      \draw      (2.65,0.5)    node[below]{$push(2)$} ;
      \draw[|->] (2.95,0.5) -- (3.8,0.5) ;
      \draw      (4,0.5)      node{$\bullet$} ;
      \draw      (3.9,0.5)      node[below]{$pop/1$} ;
      \draw[|->] (4.2,0.5) -- (4.8,0.5) ;
      \draw      (5,0.5)    node{$\bullet$} ;
      \draw      (5,0.5)    node[below]{$pop/\bot$} ;

      \draw[->,dashed] (2.95,0.6) -- (3.8,0.9) ;
      \draw[->,dashed] (4.2,0.9) -- (4.8,0.6) ;
      \draw[->,dashed] (4.2,0.6) -- (4.8,0.9) ;

    \end{tikzpicture}
    \caption{\footnotesize $\mathcal{Q}$: CC, not SC}
    \label{fig:file_2pop1_0pop2}
  \end{subfigure}
  \hspace{\fill}
  \begin{subfigure}[b]{0.39\textwidth}
    \centering
    \begin{tikzpicture}
      \draw      (5,1)      node{$\bullet$} ;
      \draw      (5,1)      node[above]{$hd/1$} ;
      \draw[|->] (5.2,1) -- (5.8,1) ;
      \draw      (6,1)      node{$\bullet$} ;
      \draw      (6,1)      node[above]{$rh(1)$} ;
      \draw[|->] (6.2,1) -- (6.8,1) ;
      \draw      (7,1)      node{$\bullet$} ;
      \draw      (7,1)      node[above]{$hd/2$} ;
      \draw[|->] (7.2,1) -- (7.8,1) ;
      \draw      (8,1)      node{$\bullet$} ;
      \draw      (8,1)      node[above]{$rh(2)$} ;

      \draw      (3,0.5)      node{$\bullet$} ;
      \draw      (3,0.5)      node[above]{$push(1)$} ;
      \draw[|->] (3.2,0.5) -- (3.55,0.5) ;
      \draw      (3.75,0.5)   node{$\bullet$} ;
      \draw      (3.75,0.5)   node[below]{$push(2)$} ;
      \draw[|->] (3.95,0.5) -- (4.8,0.5) ;
      \draw      (5,0.5)      node{$\bullet$} ;
      \draw      (5,0.5)      node[below]{$hd/1$} ;
      \draw[|->] (5.2,0.5) -- (5.8,0.5) ;
      \draw      (6,0.5)      node{$\bullet$} ;
      \draw      (6,0.5)      node[below]{$rh(1)$} ;
      \draw[|->] (6.2,0.5) -- (6.8,0.5) ;
      \draw      (7,0.5)      node{$\bullet$} ;
      \draw      (7,0.5)      node[below]{$hd/2$} ;
      \draw[|->] (7.2,0.5) -- (7.8,0.5) ;
      \draw      (8,0.5)      node{$\bullet$} ;
      \draw      (8,0.5)      node[below]{$rh(2)$} ;

      \draw[->,dashed] (3.95,0.6) -- (4.8,0.9) ;
      \draw[->,dashed] (6.2,0.9) -- (6.8,0.6) ;
      \draw[->,dashed] (6.2,0.6) -- (6.8,0.9) ;

    \end{tikzpicture}
    \caption{\footnotesize $\mathcal{Q}'$: CC, not SC}
    \label{fig:file_2pop1_0pop2_regle}
  \end{subfigure}
  \hspace{\fill}

\vspace{4mm}

  \hspace{\fill}
  \begin{subfigure}[b]{0.35\textwidth}
    \centering
    \begin{tikzpicture}
      \draw      (1,1)    node{$\bullet$} ;
      \draw      (1,1)    node[above]{$w_a(1)$} ;
      \draw[|->] (1.2,1) -- (1.8,1) ;
      \draw      (2,1)    node{$\bullet$} ;
      \draw      (2,1)    node[above]{$w_c(2)$} ;
      \draw[|->] (2.2,1) -- (2.8,1) ;
      \draw      (3,1)    node{$\bullet$} ;
      \draw      (3,1)    node[above]{$w_d(1)$} ;
      \draw[|->] (3.2,1) -- (3.8,1) ;
      \draw      (4,1)    node{$\bullet$} ;
      \draw      (4,1)    node[above]{$r_b/0$} ;
      \draw[|->] (4.2,1) -- (4.8,1) ;
      \draw      (5,1)    node{$\bullet$} ;
      \draw      (5,1)    node[above]{$r_e/1$} ;
      \draw[|->] (5.2,1) -- (5.8,1) ;
      \draw      (6,1)    node{$\bullet$} ;
      \draw      (6,1)    node[above]{$r_c/3$} ;
      
      \draw      (1,0.5)    node{$\bullet$} ;
      \draw      (1,0.5)    node[below]{$w_b(1)$} ;
      \draw[|->] (1.2,0.5) -- (1.8,0.5) ;
      \draw      (2,0.5)    node{$\bullet$} ;
      \draw      (2,0.5)    node[below]{$w_c(3)$} ;
      \draw[|->] (2.2,0.5) -- (2.8,0.5) ;
      \draw      (3,0.5)    node{$\bullet$} ;
      \draw      (3,0.5)    node[below]{$w_e(1)$} ;
      \draw[|->] (3.2,0.5) -- (3.8,0.5) ;
      \draw      (4,0.5)    node{$\bullet$} ;
      \draw      (4,0.5)    node[below]{$r_a/0$} ;
      \draw[|->] (4.2,0.5) -- (4.8,0.5) ;
      \draw      (5,0.5)    node{$\bullet$} ;
      \draw      (5,0.5)    node[below]{$r_d/1$} ;
      \draw[|->] (5.2,0.5) -- (5.8,0.5) ;
      \draw      (6,0.5)    node{$\bullet$} ;
      \draw      (6,0.5)    node[below]{$r_c/3$} ;

      \draw[->,dashed] (3.2,0.6) -- (4.8,0.9) ;
      \draw[->,dashed] (3.2,0.9) -- (4.8,0.6) ;

    \end{tikzpicture}
    \caption{\footnotesize $\mathscr{M}_{[a-z]}$: CCv but not CC}
    \label{fig:Mem_WCC_pas_CC}
  \end{subfigure}
  \hspace{\fill}
  \begin{subfigure}[b]{0.35\textwidth}
    \centering
    \begin{tikzpicture}
      \draw      (1,1)    node{$\bullet$} ;
      \draw      (1,1)    node[above]{$w_a(1)$} ;
      \draw[|->] (1.2,1) -- (1.8,1) ;
      \draw      (2,1)    node{$\bullet$} ;
      \draw      (2,1)    node[above]{$w_a(2)$} ;
      \draw[|->] (2.2,1) -- (2.8,1) ;
      \draw      (3,1)    node{$\bullet$} ;
      \draw      (3,1)    node[above]{$w_b(3)$} ;
      \draw[|->] (3.2,1) -- (3.8,1) ;
      \draw      (4,1)    node{$\bullet$} ;
      \draw      (4,1)    node[above]{$r_d/3$} ;
      \draw[|->] (4.2,1) -- (4.8,1) ;
      \draw      (5,1)    node{$\bullet$} ;
      \draw      (5,1)    node[above]{$r_c/1$} ;
      \draw[|->] (5.2,1) -- (5.8,1) ;
      \draw      (6,1)    node{$\bullet$} ;
      \draw      (6,1)    node[above]{$w_a(1)$} ;
      
      \draw      (1,0.5)    node{$\bullet$} ;
      \draw      (1,0.5)    node[below]{$w_c(1)$} ;
      \draw[|->] (1.2,0.5) -- (1.8,0.5) ;
      \draw      (2,0.5)    node{$\bullet$} ;
      \draw      (2,0.5)    node[below]{$w_c(2)$} ;
      \draw[|->] (2.2,0.5) -- (2.8,0.5) ;
      \draw      (3,0.5)    node{$\bullet$} ;
      \draw      (3,0.5)    node[below]{$w_d(3)$} ;
      \draw[|->] (3.2,0.5) -- (3.8,0.5) ;
      \draw      (4,0.5)    node{$\bullet$} ;
      \draw      (4,0.5)    node[below]{$r_b/3$} ;
      \draw[|->] (4.2,0.5) -- (4.8,0.5) ;
      \draw      (5,0.5)    node{$\bullet$} ;
      \draw      (5,0.5)    node[below]{$r_a/1$} ;
      \draw[|->] (5.2,0.5) -- (5.8,0.5) ;
      \draw      (6,0.5)    node{$\bullet$} ;
      \draw      (6,0.5)    node[below]{$w_c/1$} ;
      
      \draw[->,dashed] (3.2,0.9) -- (3.8,0.6) ;
      \draw[->,dashed] (1.3,0.9) -- (4.7,0.6) ;
      \draw[->,dashed] (3.2,0.6) -- (3.8,0.9) ;
      \draw[->,dashed] (1.3,0.6) -- (4.7,0.9) ;
    \end{tikzpicture}
    \caption{\footnotesize $\mathscr{M}_{[a-z]}$: CM but not CC}
    \label{fig:Mem_CM_pas_CC}
  \end{subfigure}
  \hspace{\fill}
  \caption{Distributed histories for instances of $\mathscr{W}_2$, $\mathscr{Q}$, $\mathscr{Q}'$ and $\mathscr{M}_{[a-z]}$ with different consistency criteria.
  }
  \label{fig:exemples_histoires}
\end{figure*}
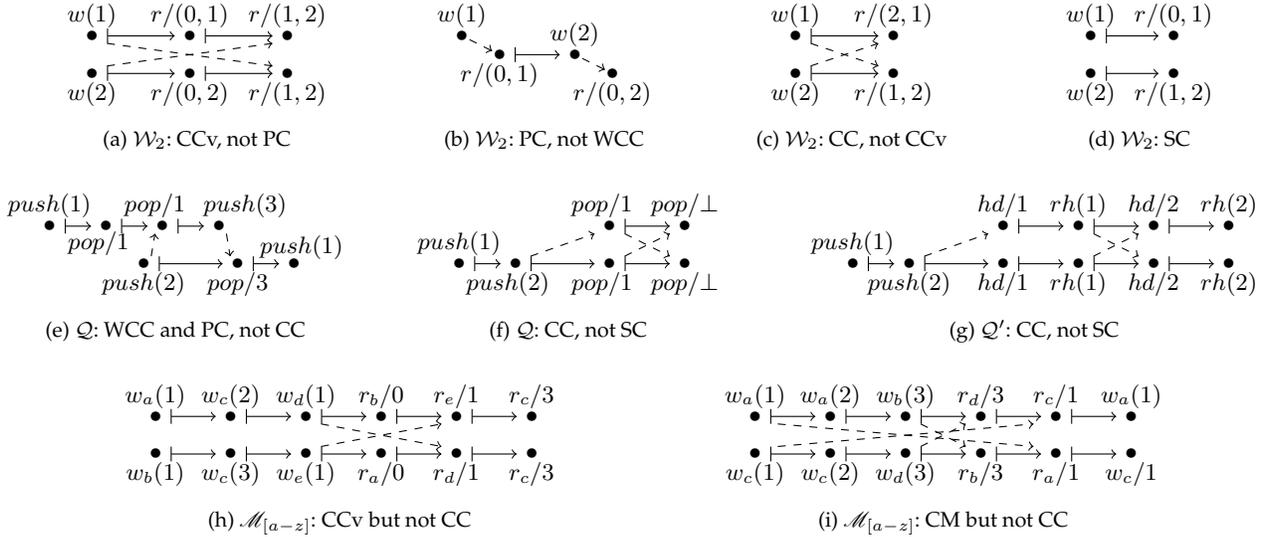

The total order of sequential consistency is a causal order that verifies two additional properties: (1) as the 
causal order is total, the concurrent present of each operation is empty and (2) the value
returned by each operation must be plausible with respect to the linearization of its causal
past (which is unique because of (1)). In our formalism, for all events $e\in E_H$, 
$\lin((H^\rightarrow).\pi(\lfloor e\rfloor, \lfloor e\rfloor))\cap L(T) \neq \emptyset$, where $\rightarrow$ 
is the causal order (Fig. \ref{subfig:CC:cones:SC}). 
Note that the existence of a causal order verifying (2) is equivalent to 
sequential consistency for infinite histories: because any concurrent 
events $e$ and $e'$ have events in common in their respective future, 
a linearization for any of these future events must order $e$ and $e'$, 
so we can build a new causal order in which $e$ and $e'$ are ordered as 
in this linearization (the complete proof is very close to the one for Proposition \ref{lem:CC_to_PC}). 
As processes cannot
know their future, any algorithm implementing (2) must also ensure (1). 
Different flavours of causal consistency that can be implemented in wait-free systems
correspond to different ways to weaken (2), as illustrated in Fig. \ref{fig:CC:cones}.

The differences between the criteria introduced in this paper are illustrated with small examples 
on instances of window streams of size 2 ($\mathcal{W}_2$), of two kinds of
queues ($\mathcal{Q}$ and $\mathcal{Q}'$) and of memory on Fig. \ref{fig:exemples_histoires}.
In these histories, the dummy values returned by update operations are ignored for the sake of clarity.
The program order is represented by solid arrows, and semantic causal relations are represented by dashed arrows 
(a read value is preceded by the corresponding write operation, a popped value needs to be pushed first, etc.).
For example, the history on Fig. \ref{fig:SC} shows two processes sharing a window stream of size 2. 
The first process first writes $1$ and then reads $(0,1)$, while the second process writes $2$ 
and then reads $(1, 2)$. As the word $w(1)/\bot.r/(0,1).w(2)/\bot.r/(1,2)$ 
is in both $\lin(H)$ and $L(\mathcal{W}_2)$, this history is sequentially consistent.

\subsection{Weak causal consistency}

Weak causal consistency precludes the situation where a process is aware of an operation 
done in response to another operation, but not of the initial operation (e.g. a question and the answer in a forum). 
In this scenario, the answer is a consequence of the question, so the
reading of the answer, that is a consequence of the question, should also be a
consequence of the question. Weak causal consistency ensures that, 
when a process performs an operation, it is aware of its whole causal past. 
In terms of time zones, the value returned by each operation
must be consistent with regard to a linearization of the side effect of 
all operations that appear in its causal past -- and only them. 
More formally, it corresponds to Def. \ref{def:weak_causal_consistency}. 
Weak causal consistency roughly corresponds to the notion of causality preservation 
in the CCI model \cite{sun1998achieving} used in collaborative editing, that requires
causality, convergence and intention 
preservation. The difference between weak causal consistency and causality preservation
stems from the fact that the model considered in this paper is based on sequential
specifications that replaces the notion of "intention" of the CCI model. 
\pagebreak

\begin{definition}\label{def:weak_causal_consistency} 
  A history $H$ is \emph{weakly causally consistent} (WCC) for an ADT $T$ if
  there exists a causal order $\rightarrow$, such that:
  $\forall e\in E_H, \lin((H^\rightarrow).\pi(\lfloor e\rfloor, \{e\}))\cap L(T) \neq \emptyset.$
\end{definition}
\vspace{1mm}

In the history on Fig. \ref{fig:PC_pas_WCC}, the operation read $r/(0,1)$ must have $w(1)/\bot$
in its his causal history for the execution to be weak causally consistent. Similarly,
$w(2)/\bot\rightarrow r/(2,1)$. The causal order of this history is total, so it
has only one possible linearization for the last read: $w(1).r.w(2).r/(2,1)$, which does not
conform to the sequential specification, thus the history is not weak causally consistent.
\newpage

On the contrary, the history of Fig. \ref{fig:WCC_pas_PC} is weak causally consistent:
$w(1)/\bot$, $w(1).w(2).r/(0,1)$, $w(1).w(2).r.r/(1,2)$, $w(2)/\bot$, $w(1).w(2).r/(0,2)$
and $w(1).w(2).r.r/(1,2)$ are correct linearizations for the six events. 
This history illustrates why pipelined consistency and eventual consistency cannot be achieved
together for all objects in wait-free message-passing systems \cite{perrin2015update} (all processes but one may crash). 
In a similar execution, a sequentially consistent window stream would verify three properties: 
(termination) all the operations must return; (validity) all the reads must return at least one non-null value; 
(agreement) the oldest value seen by each process must be the same. This problem is similar to
Consensus, that is impossible to solve in asynchronous distributed systems in the presence of process crashes \cite{fischer1985impossibility}. 
In pipelined consistency, for their second read, the first process can only return $(0,1)$ or $(1,2)$ and the second
process only $(0,2)$ or $(2,1)$; they can never converge. Pipelined consistency
sacrifices agreement to ensure termination of the first read, while eventual consistency relaxes termination 
to ensure agreement (the states will eventually be the same, but we do not know when). 

In wait-free distributed systems, pipelined consistency and eventual consistency cannot be achieved together, 
but weak causal consistency can be enriched with either pipelined consistency to form causal consistency (Sec. \ref{sec:cc})
or with eventual consistency to form causal convergence (Sec. \ref{sec:ccv}).

\subsection{Behaviour in absence of data races}

In \cite{ahamad1995causal}, causal memory is justified by the context in which it may be used. 
If a causal memory is never subject to race conditions, it behaves exactly like a sequential memory.
This is actually a property of weak causal consistency : a weakly causally consistent history 
that does not contain concurrent writes is sequentially consistent (Proposition \ref{prop:causal:WCC_W|W}).
Thus, for a program in which synchronisation does not rely on memory, a weakly causally consistent memory
ensures the same quality of service as a sequentially consistent memory with a better time efficiency.
Indeed, concurrent writes need to be synchronized to get a sequentially consistent shared memory \cite{attiya1994sequential},
but it is not necessary for weak causal consistency (see Section \ref{section:implementation}).

\begin{proposition}\label{prop:causal:WCC_W|W}
  Let $T$ be an ADT and $H = (\Sigma, E, \Lambda, \mapsto)$ be a history such that $H \in WCC(T)$ and, for all update operations $u, u' \in E$,
  $u\rightarrow u'$ or $u'\rightarrow u$. Then $H\in SC(T)$.
\end{proposition}
\begin{proof}
  Let $T$ be an ADT and $H = (\Sigma, E, \Lambda, \mapsto) \in WCC(T)$ such that, for all update operations $u, u' \in E$,
  $u\rightarrow u'$ or $u'\rightarrow u$. 

  Let $\le$ be a total order on $E$ that extends $\rightarrow$, and let $l$ be the unique linearization
  of $\lin(H^\le)$. As $\mapsto\subset\rightarrow\le$, $l\in \lin(H)$.
  Suppose that $l\notin L(T)$. As the transition system of $T$ is deterministic, there exists a finite prefix of $l$
  that does not belong to $L(T)$. Let $l'\in \Sigma^\star$ and $e\in E$ such that $l'\cdot\Lambda(e)$ is the shortest such prefix.
  As $H\in WCC(T)$, there exists a linearization  $l''\cdot\Lambda(e) \in \lin((H^\rightarrow).\pi(\lfloor e\rfloor, \{e\}))\cap L(T)$. 
  as $e$ is the maximum of $\lfloor e\rfloor$ according to $e$.
  Now, $l'$ and $l''$ are composed of the same updates in the same order, as $\rightarrow$ is total considering only the updates, 
  so $l'$ and $l''$ lead to the same state. Under these conditions, it is absurd that $l''\cdot\Lambda(e) \in L(T)$,
  $l' \in L(T)$ and $l'\cdot\Lambda(e) \notin L(T)$. It results that $l\in L(T)$, so $H\in SC(T)$.
\end{proof}

% ---------------------------------------------------------------------------------------------
\section{Causal consistency} \label{sec:cc}
% ---------------------------------------------------------------------------------------------

\subsection{Definition}

Among the four session guarantees, weak causal consistency and causal convergence
ensure \emph{Read your writes}, \emph{Monotonic writes} and \emph{Writes follows reads}, 
but not \emph{Monotonic reads} while causal consistency is supposed to ensure the four session guarantees. 
The difference between pipelined consistency and weak causal consistency 
can be understood in terms of the time zones illustrated on 
Fig. \ref{fig:CC:cones}. On the one hand, in pipelined consistency, the present must be consistent with the 
whole program past, writes as well as reads, and the writes of a prefix of the other processes, 
but there is no reference to a causal order (Fig. \ref{subfig:CC:cones:PC}). 
On the other hand, weak causal consistency focuses on causal order, but only requires consistency with the writes (Fig. \ref{subfig:CC:cones:WCC}). 
Causal consistency enforces both weak causal consistency and pipelined consistency by considering differently the 
program past and the rest of the causal past: 
the value returned by each read must respect a linearization containing all the writes of the causal history and
the reads of its program history (Fig. \ref{subfig:CC:cones:CC}). More formally, it corresponds to Def. \ref{def:causal_consistency}.

\begin{definition}\label{def:causal_consistency} 
  A history $H$ is \emph{causally consistent} (CC) for an ADT $T\in \mathscr{T}$ if
  there exists a causal order $\rightarrow$ such that:\linebreak
  $\forall p \in \mathscr{P}_H, \forall e\in p, \lin((H^\rightarrow).\pi(\lfloor e\rfloor, p))\cap L(T) \neq \emptyset.$
\end{definition}

As causal consistency is a strengthening of both pipelined consistency and weak causal consistency, 
the histories of figures \ref{fig:WCC_pas_PC} and \ref{fig:PC_pas_WCC} are not causally consistent.
On the contrary, the history of Fig. \ref{fig:SCC_pas_SC} is causally consistent:
$w(1)/\bot$, $w(2).w(1)/\bot.r/(2,1)$, $w(2)/\bot$ and $w(2)/\bot.w(1).r/(1,2)$ are linearizations 
for the four events.

Causal consistency is more than the exact addition of pipelined consistency and weak causal consistency as shown by 
Fig. \ref{fig:file_WCC_PC_pas_CC} that features a first-in-first-out queue. Several kinds of queues are 
instantiated in this paper, so their corresponding ADTs are only informally described.
In this history, the queue has two operations $push(v)$ that adds an integer value $v$ at the end of the queue, 
and $pop$ that removes and returns the first element, i.e. the oldest element pushed and not popped yet. 
This history can be interpreted for weak causal consistency: when the first process pops for the first time, it is only aware of 
its own push, so it returns $1$. When it receives the notification for the $push(2)$ operation,
it notices that value $2$ should be before value $1$ in the queue, so the first pop should have returned $2$, and
the second $1$. The linearization $push(2).push(1).pop.pop/1$ is correct for weak causal consistency. 
It is also pipelined consistent, as $push(2).pop.push(1).push(1)/\bot.pop/1.pop/1.push(3)/\bot$ and
$push(2)/\bot.push(1).pop.pop.push(3).pop/3.push(1)/\bot$ are linearizations for the two processes.
Note that the $1$ returned by the second $pop$ does not correspond to the same $push(1)$ for the two criteria.
That is why, even if the history is both pipelined consistent and weakly causally consistent, it is not causally consistent.

The history on Fig. \ref{fig:file_2pop1_0pop2} is causally consistent: both processes concurrently $pop$ the 
queue when in same state $[1, 2]$, so they both get $1$. Then they integrate the fact that the other process
removed the head, which they consider is the value $2$; at their next $pop$, the queue is empty. Weakly consistent criteria
cannot ensure that all elements inserted will be popped once and only once even if an infinity of pops are performed, 
but, this example shows that causal consistency, neither guarantees existence (2 is never popped) nor unicity (1 is popped twice). 
The reason is that, in weak consistency criteria, the transition and output parts of the operations are loosely coupled.
In Fig. \ref{fig:file_2pop1_0pop2_regle}, the pop operation is split into a $hd$ (head) operation, that returns the first element 
without removing it, and a $rh(v)$ (remove head) operation that removes the head if and only if it is equal to $v$. 
The previous pattern also may happen and both processes read $1$ and perform $rh(1)$.
However, they do not delete $2$ at the head of the queue. Using this technique, all the
values are read at least once.

The fact that causal consistency is stronger than pipelined consistency is not trivial
given the definitions: the existence of linearizations for all the events does not directly
imply the existence of a linearization for the whole history. We prove the following
proposition, that will be useful in Section \ref{sec:causal_memory}. The fact that $CC$ is
stronger than $PC$ is a direct corollary, as ${\mapsto} \subset {\rightarrow}$.

\begin{proposition}\label{lem:CC_to_PC}
  If $H$ is a causally consistent history, then \linebreak 
  $\forall p\in \mathscr{P}_H, \lin\left((H^{\rightarrow}).\pi(E_H, p)\right)\cap L(T) \neq \emptyset.$
\end{proposition}
\begin{proof}
  Let $H$ be causally consistent and $p\in \mathscr{P}_H$. 
  If $p$ is finite, it has a biggest element $e$. As $H$ is causally consistent,
  there exists a linearization  $l_e \in \lin((H^\rightarrow).\pi(\lfloor e\rfloor, p))\cap L(T)$.
  As $\mapsto \subset \rightarrow$, there exists a linearization $l$ of $(H^{\rightarrow}).\pi(E_H, p)$ whose $l_e$ is a prefix. 
  $l\in L(T)$ as $l_e\in L(T)$ and all the events that are in $l$ and not in $l_e$ are hidden.

  If $p$ is infinite, it is not possible to consider its last element. Instead, we build a growing sequence $(l_k)$
  of linearizations that converges to the whole history. The successive linearizations of the events are not
  necessarily prefixes of each other, so the linearizations we build also contain 
  a part of the concurrent present. We number the events of $p$ by $e_1\mapsto e_2 \mapsto ...$ and we define, 
  for all $k$, the set $L_k$ such that $l.e_k\in L_k$ if and only if it can be completed, by a word $l'$ such that \linebreak
  $l.e_k.l'\in \lin((H^\rightarrow).\pi(\{e\in E_H : e_k \not\rightarrow e\}, p))\cap L(T)$. In other words, 
  $L_k$ contains the linearizations of the causal past and the concurrent present of $e_k$, truncated to $e_k$.
  As $L_k$ contains the correct linearizations for causal consistency, it is not empty. It is also finite because
  $\rightarrow$ is a causal order, so $E_H\setminus \{e\in E_H : e_k \rightarrow e\}$ is finite. 
  Notice that all the linearizations in $L_{k+1}$ have a prefix in $L_k$ as $e_k\rightarrow e_{k+1}$ and $L(T)$ 
  is closed by prefixing.

  As $L_k$ is finite for all $k$ and all $l_j\in L_j$ has a prefix in $L_k$ for $j\ge k$, there is a $l_k$ that
  is the prefix of a $l_j$ for all $j\ge k$. We can build by induction a sequence $(l_k)_{k\in \mathbb{N}}$ of words of $L(T)$ such that for all $k$, $l_k\in L_k$
  and $l_{k+1}$ is a prefix of $l_k$. The sequence $(l_k)$ converges to an infinite word $l$.
  All the prefixes of $l$ are in $L(T)$, so $l\in L(T)$. Moreover, $l$ contains all the events of $H.\pi(E_H, p)$ because
  $\rightarrow$ is a causal order (so all events are in the causal history of a $e_k$ for some $k$), and 
  the causal order is respected for each pair of events, because it is respected by all the prefixes of $l$
  that contain those two events. Finally, $l\in \lin\left((H^{\rightarrow}).\pi(E_H, p)\right)\cap L(T)$.
\end{proof}

% ---------------------------------------------------------------------------------------------
\subsection{Causal consistency versus causal memory}\label{sec:causal_memory}
% ---------------------------------------------------------------------------------------------

Memory is a particular abstract data type; as an example, causal memory has been defined in \cite{ahamad1995causal}.
In this section, we compare causal consistency applied to memory and causal memory. We first 
recall the formal definitions of memory and causal memory, then we exhibit a difference between the two associated consistency criteria when the same value is written twice in the same register (a register being a piece of memory).
We finally prove that, when all the values written are different, causally consistent memory corresponds exactly
to causal memory.

\pagebreak
We now define memory as an abstract data type. A memory is a pool of integer registers. 
As causal consistency is not composable, it is important to define a causal memory 
as a causally consistent pool of registers rather than a pool of causally consistent registers, which is very different.
An integer \emph{register} $x$ is isomorphic to a window stream of size 1. 
It can be accessed by a \emph{write} operation $w_x(v)$, where $v\in \mathbb{N}$ is the written value and 
a \emph{read} operation $r_x$ that returns the last value written, if there is one, 
or the default value $0$ otherwise. The integer \emph{memory} $\mathcal{M}_X$ is 
the collection of the integer registers of $X$. More formally, it corresponds
to the ADT given in Def. \ref{def:mem}. 
In all the section, let $\mathcal{M}_X$ be a memory abstract data type. 

\begin{definition}\label{def:mem}
  Let $X$ be any set of symbolic \emph{register names}. We define the integer \emph{memory} on $X$ by the ADT \linebreak
  $\mathcal{M}_x = (\Sigma_i, \Sigma_0, Q, q_0, \delta, \lambda)$ with $Q = \mathbb{N}^X$, $q_0 : x\mapsto 0$,
  $\Sigma_i = \{r_x, w_x(v) : v\in \mathbb{N}, x\in X\}$, 
  $\Sigma_o = \mathbb{N} \cup \{\bot\}$,  
  and for all $x\neq y\in X$, $v\in \mathbb{N}$ and $q\in X\rightarrow\mathbb{N}$,
  $\delta(q, w_x(v))(x) = v$, $\delta(q, w_x(v))(y) = q(y)$, $\lambda(q, w_x(v)) = \bot$, $\delta(q, r_x) = q$ and $\lambda(q, r_x) = q(x)$.
\end{definition}

The dichotomy between causal consistency and causal convergence also exists for memory. 
On Fig. \ref{fig:Mem_WCC_pas_CC}, assuming the first read of each process only has the writes of
the same process in their causal past, all the writes of the other processes must be placed after 
this read. In order to satisfy causal consistency, the register $c$ must be set to $3$ for the 
first register and to $2$ for the second register in the end, which cannot be reconciled with 
causal convergence.

Causal memory defines a causal order explicitly from the history by considering the reads and the writes. 
This causal order has the same use as the program order in pipelined consistency. More formally,
it corresponds to Def. \ref{def:CM}.

\begin{definition}\label{def:CM}
  A relation $\rightsquigarrow$ is a writes-into order if:
  \begin{itemize}
  \item for all $e, e'\in E_H$ such that $e \rightsquigarrow e'$, there are $x\in X$ and $v\in \mathbb{N}$ such that
    $\Lambda(e) = \mathrm{w}_x(v)$ and $\Lambda(e') = \mathrm{r}_x/v$,
  \item for all $e\in E_H$, $|\{e'\in E_H : e' \rightsquigarrow e\}| \le 1$,
  \item for all $e\in E_H$ such that $\Lambda(e) = \mathrm{r}_x/v$ and there is no $e'\in E_H$ such that $e'\rightsquigarrow e$,
    then $v = 0$.
  \end{itemize}
  A history $H$ is $\mathcal{M}_X$-\emph{causal} (CM) if there exists a writes-into order $\rightsquigarrow$ such that:
  \begin{itemize}
  \item there is a causal order $\rightarrow$ that contains $\rightsquigarrow$ and $\mapsto$,
  \item $\forall p\in \mathscr{P}_H, \lin\left((H^\rightarrow).\pi(E_H, p)\right) \cap L(\mathcal{M}_X) \neq \emptyset$.
  \end{itemize}
\end{definition}

Causal consistency and causal memory are not identical. 
This comes from the fact that the writes-into order is not unique.
This weakens the role of the logical time, as the intuition that a read must be bound to its 
corresponding write is not always captured by the definition.
Let us illustrate this point with the history on Fig. \ref{fig:Mem_CM_pas_CC}.
In this history, we consider the writes-into order in which the reads on $x$ and $z$ 
are related to the first write of the other process. This writes-into order is correct,
as each read is related to exactly one write, and the registers and the values are the same. Moreover, 
the linearizations $w_a(1)/\bot.w_a(2)/\bot.w_b(3)/\bot.w_c(1).w_c(2).w_d(3).$\linebreak $r_d/3.r_b.r_a.w_c(1).r_c/1.w_a(1)/\bot$
and $w_a(1).w_a(2).w_b(3).$\linebreak $w_c(1)/\bot.w_c(2)/\bot.w_d(3)/\bot.r_b/3.r_d.r_c.w_a(1).r_a/1.w_c(1)/\bot$ for the two processes
are correct, so this history is correct for causal memory.
However, in these linearizations, the value read by the two last reads was not written by their predecessors 
in the writes-into relation. If we change this relation to restore the real data dependencies, we obtain a cycle
in the causal order. This example shows that the approach of Def. \ref{def:CM}, that uses the semantics of the operations,
is not well suited to define the consistency criteria. 

This issue is usually solved \cite{misra1986axioms} by the hypothesis that all written values are distinct.
Even if this can be achieved by the addition of unique timestamps on the values stored in the memory, this solution 
is not acceptable because it changes the way the final object can be used. We now prove that, under this hypothesis,
causal consistency and causal memory are equal. It means that causal consistency solves the problem raised above,
while remaining as close as possible to the original criterion.

\begin{proposition}
  Let $H$ be a distributed history.  If $H \in CC(\mathcal{M}_X)$, then $H$ is $\mathcal{M}_X$-causal.
\end{proposition}
\begin{proof}
  Suppose $H$ is causally consistent. For each event $e$, there exists a process $p_e$ with $e\in p_e$ 
  and a linearization $l_e \in \lin((H^\rightarrow).\pi(\lfloor e\rfloor, p_e))\cap L(\mathcal{M}_X)$.
  Note that these processes and linearizations are not necessarily unique, but we fix them for each event now. 
  Let us define the writes-into order $\rightsquigarrow$ by $e\rightsquigarrow e'$ if $\Lambda(e') = \mathrm{r}_x/v$
  and $e$ is the last write on $x$ in $l_e$. As $l_e\in L(\mathcal{M}_X)$, $\Lambda(e) = \mathrm{w}_x(v)$.
  $e'$ also has at most one antecedent by $\rightsquigarrow$, and if it has none, then $v=0$. 
  The transitive closure $\xrightarrow{CM}$ of $\rightsquigarrow \cup \mapsto$ is a partial order contained into
  $\rightarrow$. By Proposition \ref{lem:CC_to_PC}, for all $p\in \mathscr{P}_H$,
  $\lin\left((H^\rightarrow).\pi(E_H, p)\right) \cap L(\mathcal{M}_X) \neq \emptyset$, so  
  $H$ is $\mathcal{M}_X$-causal.
\end{proof}

\begin{proposition}
  Let $H$ be a distributed history such that, for all $e\neq e'\in E_H$ with $\Lambda(e) = w_x(v)/\bot$ and 
  $\Lambda(e') = w_y(v')/\bot$, $(x,v) \neq (y,v')$. If $H$ is $\mathcal{M}_X$-causal, then $H \in CC(\mathcal{M}_X)$.
\end{proposition}
\begin{proof}
  Suppose that $H$ is $\mathcal{M}_X$-causal. $\xrightarrow{CM}$ is a causal order.
  Let $p\in \mathscr{P}_H$ and $e\in p$. There exists a linearization 
  $l_p \in \lin((H^{\xrightarrow{CM}}).\pi(E_H, p)) \cap L(\mathcal{M}_X)$, associated with a total order of the events 
  $\le_p$. 
  Let $l_e$ be the unique linearization of $\lin((H^{\le_p}).\pi(\lfloor e\rfloor), p)$. 
  Let $e'\in \lfloor e\rfloor$ labelled by $r_x/v$. If $e'$ has no antecedent in the writes-into order, $v=0$. Otherwise, 
  this antecedent $e''$ is the last write on $x$ before $e'$ in $l_p$, because it is the only event 
  labelled $w_x(v)$ in the whole history. As $e'' \xrightarrow{CM} e'$, it is also the last write on $x$ before $e'$ in $l_e$.
  All in all, $l_e \in L(\mathcal{M}_X)$ and $H$ is causally consistent.
\end{proof}

% ---------------------------------------------------------------------------------------------
\section{Causal convergence} \label{sec:ccv}
% ---------------------------------------------------------------------------------------------

\subsection{Definition}

Eventual consistency \cite{vogels2008eventually} requires that, if at one point, 
all the processes stop doing updates (i.e. operations with a side effect), 
then eventually, all local copies of the object will converge to a common state.

Causal convergence assumes weak causal consistency and eventual consistency. It strengthens weak causal consistency by imposing that the linearizations obtained for all the events correspond to the same total order. Consequently, in causal convergence, the updates are totally ordered and the state read by each operation is the result of
the updates in its causal past, ordered by this common total order. Thus, two operations with the 
same causal past are done in the same state.

\begin{definition}\label{def:causal_convergence} 
  A history $H$ is \emph{causally convergent} (CCv) for an ADT $T$ if there exists a causal 
  order $\rightarrow$, and a total order $\le$ that contains $\rightarrow$ such that:\\
  $\forall e\in E_H, \lin((H^\le).\pi(\lfloor e\rfloor, \{e\}))\cap L(T) \neq \emptyset.$
\end{definition}

\pagebreak 
The history on Fig. \ref{fig:WCC_pas_PC} is causally convergent: the causal order and the
linearizations introduced in Section \ref{section:weak_causal_consistency} could be obtained considering any total
order $\le$ in which \linebreak $w(1)/\bot \le w(2)/\bot$. The history on Fig. \ref{fig:SCC_pas_SC}, yet, 
is not causally convergent: both writes must be in the causal past of both
reads as both values are read, but they were not applied in the same order.

A consistency criterion called strong update consistency has been
introduced in \cite{perrin2015update} as a strengthening of both update
consistency and strong eventual consistency \cite{BGYZ14},
that both strengthen eventual consistency. It is interesting to observe that 
causal convergence is stronger than strong update consistency, as it 
imposes to the visibility relation to be a transitive causal order. In 
other words, there is the same relation between strong update 
consistency and causal convergence as between pipelined consistency and 
causal consistency.

\subsection{Behaviour in absence of data races}

Because causal convergence is stronger than weak causal consistency, Proposition \ref{prop:causal:WCC_W|W} also applies to it. 
Besides it, there is another situation in which causal convergence behaves like sequential consistency: Proposition \ref{prop:causal:CCv_W|R}
proves that causally convergent histories in which no updates happen concurrently to queries are also sequentially consistent. 

\begin{proposition}\label{prop:causal:CCv_W|R}
  Let $T$ be an ADT and $H = (\Sigma, E, \Lambda, \mapsto)$ be a concurrent history such that $H \in WCC(T)$ and, 
  for all update operations $u \in E$
  and query operations $q \in E$, $u\rightarrow q$ or $q\rightarrow u$. Then $H\in SC(T)$.
\end{proposition}

\begin{proof}
  Let $T$ be an ADT and $H = (\Sigma, E, \Lambda, \mapsto) \in WCC(T)$ such that, for all update operations $u \in E$
  and query operations $q \in E$, $u\rightarrow q$ or $q\rightarrow u$. As $H\in CCv(T)$, there exists a total order $\le$ 
  that contains $\rightarrow$ and, for all $e\in E$, a linearization 
  $l_e\cdot\Lambda(e) \in \lin((H^\le).\pi(\lfloor e\rfloor, \{e\})\cap L(T)$.

  Let $l$ be the unique linearization of $\lin(H^\le)$. As $\mapsto \subset \le$, $l\in \lin(H)$. 
  Suppose that $l\notin L(T)$. As in Proposition \ref{prop:causal:WCC_W|W},
  $l$ has a prefix $l'\cdot\Lambda(e) \notin L(T)$ with $l'\in L(T)$. 
  As the transition system of $T$ is complete, $e$ can not be a pure update.
  It means $e$ is a query operation, so it is not concurrent with an update operation. 
  As $\rightarrow\subset\le$, $e$ has the same updates in its causal past and in its 
  predecessors by $\le$, which means that $l_e$ and $l'$ are composed of the same 
  updates in the same order, so $l'$ and $l''$ lead to the same state. 
  Under these conditions, it is absurd that $l''\cdot\Lambda(e) \in L(T)$,
  $l' \in L(T)$ and $l'\cdot\Lambda(e) \notin L(T)$. It results that $l\in L(T)$, so $H\in SC(T)$.
\end{proof}

% ============================================================================
\section{Implementation in wait-free systems}\label{section:implementation}
% ============================================================================

In this section, we illustrate how causally consistent data structures can be implemented 
in a special kind of distributed systems: wait-free asynchronous message-passing distributed systems.
We first introduce our computing model, then we give an implementation of an array of $K$ 
window streams of size $k$ for causal consistency and causal convergence. 

\subsection{Wait-free asynchronous message-passing distributed systems}

A message-passing distributed system is composed of a known number $n$ of sequential processes
that communicate by sending and receiving messages. Processes are asynchronous.
This means that the processes execute each at its own pace, and there is no bound on the time between 
the sending and the reception of a message. Moreover, processes can crash. A process that crashes 
simply stops operating. A process that never crashes is said to be \emph{non-faulty}.

Communication is done by the mean of a reliable causal broadcast communication primitive \cite{HT93}. 
Such a communication primitive can be implemented on any system where eventually reliable point-to-point 
communication is possible. Processes can use two operations \emph{broadcast} and \emph{receive} with the 
following properties: 
\begin{itemize}
\item If a process receives a message $m$, then $m$ was broadcast by some process.
\item If a process receives a message $m$, then all non-faulty processes eventually receive $m$;
\item When a non-faulty process broadcasts a message, this message is immediately received locally at this process.
\item If a process broadcasts a message $m$ after receiving a message $m'$ then no process
receives $m$ before $m'$.
\end{itemize}

We make no assumption on the number of crashes that can occur during one execution.
In such a context a process cannot wait for the contribution of other processes 
without risking to remain blocked forever. Consequently, the execution speed of a 
process does not depend on other processes or underlying communication delays, hence 
the name "wait-free" we give to this system. Wait-free asynchronous message-passing 
distributed systems are a good abstraction of systems where the synchronisation of
is impossible (e.g. clouds where partitions can occur) or too costly 
(e.g. high performance parallel computing where synchronisation is a limitation to performances).

At the shared objects level, processes invoke operations on shared objects. These calls entail the execution of the algorithms corresponding to their message-passing implementation based on the causal reliable broadcast. An execution of a program using an object $T$ is represented (at the shared object level) by a concurrent history $(\Sigma, E, \Lambda, \mapsto)$ where $\Sigma$ is defined as in $L(T)$, $E$ is the set of all the calls to operations of $T$ during the whole execution, an event $e\in E$ is labelled by $\Lambda(e) = \sigma_i/\sigma_o$ if $\sigma_i$ is the input symbol of the operation called by $e$ and $\sigma_o$ is the return value. For $e, e'\in E$, $e\mapsto e'$ if $e$ happened before $e'$, on the same process.

\subsection{Implementation of causal consistency}

Algorithm given in Fig. \ref{algo:CC} shows an implementation of a causally consistent array of $K$ 
window streams of size $k$. The algorithm provides the user with two primitives,
$read(x)$, where $x<K$ is a stream identifier, that corresponds to a call to operation $r$
on the $x^{\text{th}}$ stream of the array, and $write(x, v)$, that corresponds to a call 
to operation $w(v)$ on the $x^{\text{th}}$ stream of the array.

Process $p_i$ maintains one variable $str_i$ that reflects the local state of the $K$ window streams.
When $p_i$ wants to read a stream, it simply returns the corresponding local state.
To write a value $v$ in a stream $x$, $p_i$ causally broadcasts a message composed of $x$ and $v$. 
Upon the reception of such a message, a process applies the writes locally by shifting the old values and 
inserting the new value at the end of the stream. 

Whenever a read or write operation is issued, it is completed without waiting for any other process. 
This corresponds to wait-free executions in shared memory distributed systems and implies fault-tolerance.

\begin{algorithm}
  \SetKw{Var}{var}
  \SetKw{Fun}{fun}
  \SetKw{Receive}{on receive}
  \SetKw{Broadcast}{causal\_broadcast}

  \SetKwData{Kwstate}{str}
  \SetKwData{Kwx}{x}
  \SetKwData{Kwv}{v}

  \SetKwFunction{KwPid}{pid}
  
  \SubAlgo{\textbf{object} $CC(\mathscr{W}_k^K)$}{

    \Var $\Kwstate_i \in \mathbb{N}^{K\times k} \leftarrow [[0,...,0],...,[0,...,0]]$\;

    \SetKwFunction{Kwread}{read}
    \SubAlgo{\Fun \Kwread $(\Kwx \in [0, K[) \in \mathbb{N}^k$}{
      \Return $\Kwstate_i[\Kwx]$\;
    }
    \SetKwFunction{Kwwrite}{write}
    \SetKwFunction{KwMessageA}{Mess}
    \SubAlgo{\Fun \Kwwrite $(\Kwx \in [0, K[, \Kwv \in \mathbb{N})$}{
      \Broadcast \KwMessageA $(\Kwx, \Kwv)$\label{algo:CC:brodcast}\;
    }
    \SubAlgo{\Receive \KwMessageA $(\Kwx \in [0, K[, \Kwv \in \mathbb{N})$}{
      \For{$y \in [0, k-2]$}{
        $\Kwstate_i[\Kwx][y] \leftarrow \Kwstate_i[\Kwx][y+1]$\;
      }
      $\Kwstate_i[\Kwx][k-1] \leftarrow \Kwv$\;
    }
  }
  \caption{Implementation of causal consistency for an \\ array of $K$ window streams of size $k$ (code for $p_i$)}
  \label{algo:CC}
\end{algorithm}

\begin{proposition}\label{prop:causal:impl_CC}
  All histories admitted by the algorithm of Fig. \ref{algo:CC} are causally consistent for the array of $K$ window streams of size $k$.
\end{proposition}
\begin{proof}
  Let $H=(\Sigma, E, \Lambda, \mapsto)$ be a history admitted by algorithm of Fig. \ref{algo:CC}. 
  For two processes $p$ and $p'$ and an event $e\in E$ invoked by $p'$, 
  we define the time $t_e^p$ as the maximum between the termination time of $e$ and the
  moment when $p$ has received all the messages sent by $p'$ before the termination of $e$.
  We use it to define two kinds of relations: 
  \begin{itemize}
  \item A causal order $\rightarrow$ by, for any two events $e, e'\in E$ invoked by processes 
    $p$ and $p'$ respectively, $e\rightarrow e'$ if $e=e'$ or $t_e^{p'}$ is a time before the beginning 
    of $e'$;
  \item For each process $p$, a total order odrer $\le_p$ by, for all events $e, e'\in E$, 
    $e\le_p e'$ if $t_e^{p}\le t_{e'}^{p}$.
  \end{itemize}
  The causal order $\rightarrow$ is reflexive by definition, antisymmetric because it is contained into the interval order defined by real-time
  and transitive because the broadcast is causal. As the messages are received instantly by their transmitter, $\rightarrow$
  contains $\mapsto$, and as messages are eventually received by all non-faulty processes, $\rightarrow$ is a causal order. 
  The relation $\le_p$ is a total order that contains $\rightarrow$.

  Let $p \in \mathscr{P}_H$, $e\in p$ and $l_p$ the unique linearization of $(H^{\le_p}).\pi(E, p)$. As $str_i$ is
  only modified when a message is received and the only reads we consider are those of $p$, $l_p\in L(\mathscr{W}_k^K)$. 
  For all $e\in p$, the prefix $l_e$ of $l_p$ until $e$ is in $\lin((H^\rightarrow).\pi(\lfloor e\rfloor, p))\cap L(\mathscr{W}_k^K)$,
  so $H\in CC(\mathscr{W}_k^K)$. 
\end{proof}

For shared memory, it is well-known \cite{gambhire2000reducing} 
that causal reception implements a little more than causality. The same thing happens for other kinds of objects. 
For example, the history on Fig. \ref{fig:SCC_pas_SC} presents an example of false causality: 
this history is causally consistent but is not admitted by the algorithm of Fig. \ref{algo:CC}. Indeed, at least 
one of the messages sent during the execution of each of the events $w_x(1)$ and $w_x(2)$ 
must be received by one of the processes after the second write has 
been enforced locally -- otherwise each of the events would precede the other in the happened-before relation.
As a result, it is not possible that both processes read the value they proposed before the other value.
Actually, the algorithm of Fig. \ref{algo:CC} ensures a slightly stronger property: for each process $p$, the linearizations 
required by causal consistency for successive events are prefix one from another. 

\subsection{Implementation of causal convergence}

Eventual consistency received such interest because it can be wait-free implemented independently 
from communication delays and moreover, it allows processes to share \textit{in fine} the same final state. 
Causal convergence can thus be seen as an improvement of eventual consistency as it ensures 
stronger consistency properties still wait-free implementable and hence weaker that sequential consistency.

The algorithm given in Fig. \ref{algo:CCv} shows an implementation of a causally convergent array of $K$ 
window streams of size $k$. The algorithm provides the user with the same interface as the algorithm of Fig.  \ref{algo:CC}:
a primitive $read(x)$ that corresponds to an operation $r$ on the $x^{\text{th}}$ stream and
a primitive $write(x, v)$, that corresponds to an operation $w(v)$ on the $x^{\text{th}}$ stream of the array.

The principle is to build a total order on the write operations on which all the participants agree, 
and to sort the corresponding values in the local state of each process with respect to this total order. 
In the algorithm of Fig. \ref{algo:CCv}, this order is built from a Lamport's clock \cite{lamport1978time} that contains 
the happened-before precedence relation, and thus is compatible with the causal order we build on top of this relation. 
A logical Lamport's clock is a pre-total order as some events may be associated with the same logical time. 
In order to have a total order, the writes are timestamped with a pair composed of the logical time and the 
id of the process that produced it (process ids are assumed unique and totally ordered). 

Process $p_i$ maintains two variables: a Lamport clock $vtime_i$ and an array $str_i$. 
Each cell of $str_i$ is an array of size $k$ that encodes the $k$ values of a window stream as
structures $(v, (vt, j))$ where $v$ is the value itself, and $(vt, j)$ is the timestamp of the 
write event that proposed it. Timestamps can be compared : $(vt, j) < (vt', j')$ if $vt < vt'$
or $vt = vt'$ and $j < j'$. With zeach write operation is associated a virtual time greater than $1$, so the 
timestamps $(0,0)$, present in the initial value of $str_i$, are smaller than the timestamps 
of all the writes. 

When $p_i$ wants to read a stream, it removes the timestamps from the corresponding local state.
To write a value $v$ in a stream $x$, $p_i$ causally broadcasts a message composed of $x$ and $v$ and a new timestamp $(vt, i)$. 
At reception of such a message, it increments its variable $vtime_i$ to implement the logical time, and
it inserts the new value at its correct location in the corresponding stream. 

\begin{algorithm}
  \SetKw{Var}{var}
  \SetKw{Fun}{fun}
  \SetKw{Receive}{on receive}
  \SetKw{Broadcast}{causal\_broadcast}

  \SetKwData{Kwclock}{vtime}
  \SetKwData{Kwstate}{str}
  \SetKwData{Kwx}{x}
  \SetKwData{Kwv}{v}
  \SetKwData{Kwcl}{vt}
  \SetKwData{Kwj}{j}

  \SetKwFunction{KwPid}{pid}
  
  \SubAlgo{\textbf{object} $CCv(\mathscr{W}_k^K)$}{

    \Var $\Kwstate_i \in \mathbb{N}^{K\times k \times (1+2)} \leftarrow [[[0,(0,0)],...],...]$\; 
    \Var $\Kwclock_i \in \mathbb{N} \leftarrow 0$\; 

    \SetKwFunction{Kwread}{read}
    \SubAlgo{\Fun \Kwread $(\Kwx \in [0, K[) \in \mathbb{N}^k$}{
      \Return $[\Kwstate_i[\Kwx][0][0], ..., \Kwstate_i[\Kwx][k-1][0]]$\;
    }
    \SetKwFunction{Kwwrite}{write}
    \SetKwFunction{KwMessageA}{Mess}
    \SubAlgo{\Fun \Kwwrite $(\Kwx \in [0, K[, \Kwv \in \mathbb{N})$}{
      \Broadcast \KwMessageA $(\Kwx, \Kwv, \Kwclock + 1, i)$  \label{algo:CCv:bc}\;
    }
    \SubAlgo{\Receive \KwMessageA $(\Kwx \in [0, K[, \Kwv \in \mathbb{N}, \Kwcl\in \mathbb{N}, \Kwj\in\mathbb{N})$}{
      $\Kwclock_i\leftarrow \max(\Kwclock_i, \Kwcl)$\label{algo:CCv:max}\;
      \Var $y\in\mathbb{N} \leftarrow 0$\;
      \While{$y < k-1 \land \Kwstate_i[\Kwx][y][1] \le (\Kwcl,\Kwj)$}{
        $\Kwstate_i[\Kwx][y] \leftarrow \Kwstate_i[\Kwx][y+1]$\;
        $y \leftarrow y+1$\;
      }
      \If{$y \neq 0$}{
        $\Kwstate_i[\Kwx][y-1] \leftarrow \Kwv$
      }
    }
  }
  \caption{Implementation of causal convergence for an array of $K$ window streams of size $k$ (code for $p_i$)}
  \label{algo:CCv}
\end{algorithm}

\begin{proposition}\label{prop:causal:impl_CCv}
  All histories admitted by Algo. of Fig. \ref{algo:CCv} are causally convergent for the array of $K$ objects $\mathscr{W}_k$.
\end{proposition}
\begin{proof}
  Let $H=(\Sigma, E, \Lambda, \mapsto)$ be a history admitted by the algoritm of Fig. \ref{algo:CCv}. 
  We define the causal order as in Proposition \ref{prop:causal:impl_CC} and a total order $\le_w$
  on the write operations by, for all writes $e_i, e_j$ invoked by $p_i$ and $p_j$ when their 
  variable $vtime$ is equal to $vt_i$ and $vt_j$, $e\le e'$ if $(vt_i, i) \le (vt_j,j)$. 
  Let us remark that, thanks to lines \ref{algo:CCv:bc} and \ref{algo:CCv:max}, 
  if $e\rightarrow e'$, then $e\le_w e'$, which means $\rightarrow\cup \le_w$ is a 
  partial order on $E$ that can be extended into a total order $\le$.
  Let $e \in E$ and $l_e$ be the unique linearization of $(H^\le).\pi(\lfloor e\rfloor, \{e\})$. If $e$ is a write,
  $l_e$ contains no return value, so $l_e\in L(T)$. Otherwise, $\Lambda(e) = read(x)/[v_0, ..., v_{k-1}]$, where,
  by construction at the reception of the messages, $v_0, ..., v_{k-1}$ are the $k$ newest values written on $x$ 
  with respect to the order $\le$, which means $l_e\in L(T)$. Finally, $H\in CCv(\mathscr{W}_k^K)$. 
\end{proof}

% ============================================================================
\section{Conclusion}\label{section:conclusion}
% ============================================================================

Sharing objects is essential to abstract communication complexity 
in large scale distributed systems. Until now, a lot of work has been done
to specify many kinds of shared memory, but they cannot always be easily 
extended to other abstract data types. In this paper, we extend causal 
consistency to all abstract data types. We also explore the variations
of causal consistency, around three consistency criteria.

Each of these three consistency criteria is pertinent. Weak causal consistency can be seen as the causal common denominator of the two branches of weak consistency criteria (eventual consistency and pipelined consistency). Indeed, it can be combined with any of them in wait-free distributed systems. Causal convergence is the result of the integration on weak causal consistency in the eventual consistency branch. Finally, causal consistency is from one side the generalization of the consistency criterion of causal memory to any abstract data type and on the other side it covers both weak causal consistency and pipelined consistency. 

To sum up, this paper allows to better understand the connections between weak consistency criteria. The two criteria to keep in mind are causal convergence and causal consistency as representatives of the two irreconcilable branches of consistency in wait-free distributed systems (sequential consistency and linearizability cannot be implemented in such a context). 

\section*{Acknowledgment}\label{section:ack}
\addcontentsline{toc}{section}{Acknowledgment} 

This work has been partially supported by the  French ANR project Socioplug (ANR-13-INFR-0003),
which is devoted to large scale distributed programming, and the Franco-German ANR project DISCMAT 
devoted to connections between mathematics  and distributed computing.

\bibliographystyle{plain}
\bibliography{ppopp}

\end{document}